\documentclass[journal]{IEEEtran}
\usepackage{epsfig}
\usepackage{amsmath}
\newtheorem{theorem}{Theorem}[section]

\newtheorem{lemma}{Lemma}[section]

\newtheorem{corollary}[theorem]{Corollary}

\newtheorem{definition}{Definition}[section]

	\def\maxproblemsmall#1#2#3#4{\fbox
		 {\begin{tabular*}{0.47\textwidth}
			{@{}l@{\extracolsep{\fill}}l@{\extracolsep{6pt}}l@{\extracolsep{\fill}}c@{}}
				#1 & $\maximize{#2}$ & $#3$ & $ $ \\[5pt]
					  $\subject\ $ &   & $#4$ & $ $
			\end{tabular*}}
			}
		\def\problemsmall#1#2#3#4{\fbox
		 {\begin{tabular*}{0.47\textwidth}
			{@{}l@{\extracolsep{\fill}}l@{\extracolsep{6pt}}l@{\extracolsep{\fill}}c@{}}
				#1 & $\minimize{#2}$ & $#3$ & $ $ \\[5pt]
					  $\subject\ $ &    & $#4$ & $ $
			\end{tabular*}}
			}

\usepackage{amsmath, amssymb, xspace}
\usepackage{epsfig}
\usepackage{longtable}
\usepackage{color}
\usepackage{mathrsfs}
\usepackage{subfig}

\def\Zbb{{\mathbb{Z}}}

\def\Nbb{{\mathbb{N}}}
\def\Fbb{{\mathbb{F}}}

\def\bfone{{\bf 1}}

\def\t{^\top}
\def\th{^{\rm th}}

\def\bkE{{\rm I\kern-.17em E}}
\def\bk1{{\rm 1\kern-.17em l}}
\def\bkD{{\rm I\kern-.17em D}}
\def\bkR{{\rm I\kern-.17em R}}
\def\bkP{{\rm I\kern-.17em P}}

\def\bkZ{{\bf{Z}}}
\def\bkE{{\rm I\kern-.17em E}}
\def\bk1{{\rm 1\kern-.17em l}}
\def\bkD{{\rm I\kern-.17em D}}
\def\bkR{{\rm I\kern-.17em R}}
\def\bkP{{\rm I\kern-.17em P}}
		\def\bkE{{\rm I\kern-.17em E}}
		\def\bk1{{\rm 1\kern-.17em l}}
		\def\bkD{{\rm I\kern-.17em D}}
		\def\bkR{{\rm I\kern-.17em R}}
		\def\bkP{{\rm I\kern-.17em P}}
		\def\bkY{{\bf \kern-.17em Y}}
		\def\bkZ{{\bf \kern-.17em Z}}
		\def\bkC{{\bf  \kern-.17em C}}

\def\Escr{\mathcal{E}}

\def\Dsf{{\sf D}}
\def\Isf{{\sf I}}
\def\Nsf{{\sf N}}
\def\Csf{{\sf C}}

\def\Hscr{{\cal H}}

\def\Cscr{{\cal C}}

\def\where{\;\textrm{where}\;}
\def\eef{\;\textrm{if}\;}
\def\ow{\;\textrm{otherwise}}

\def\non{\nonumber}

\let\forallnew\forall
\renewcommand{\forall}{\forallnew\ }
\let\forall\forallnew

\def\ds{\displaystyle}

\def\b12{(\beta_1,\beta_2)}

\newcommand{\I}[1]{\mathbb{I}_{\{#1\}}}

\def\subject{\hbox{\rm subject to}}
\newcommand{\Real}{\ensuremath{\mathbb{R}}}

\def\limn{\ds \lim_{n \rightarrow \infty}}

\def\exp{\mathop{\hbox{\rm exp}}}

\def\half  {{\textstyle{1\over 2}}}

\def\superstar{^{\raise 0.5pt\hbox{$\nthinsp *$}}}
\def\SUPERSTAR{^{\raise 0.5pt\hbox{$*$}}}
\def\lamstarT {\lambda^{\raise 0.5pt\hbox{$\nthinsp *$}T}}

\newlength{\noteWidth}
\long\def\notes#1{\ifinner
{\tiny #1}
\else
\marginpar{\parbox[t]{\noteWidth}{\raggedright\tiny #1}}
\fi\typeout{#1}}
 \def\notes#1{\typeout{read notes: #1}} %uncomment for final version

\newcommand{\ie}{i.e.\@\xspace} %%% i.e.,
\newcommand{\eg}{e.g.\@\xspace} %%% e.g.,
\newcommand{\etal}{et al.\@\xspace} %%% e.g., Gill \etal (1986)
\newcommand{\Matlab}{\textsc{Matlab}\xspace}
\newcommand{\minimize}[1]{\displaystyle\minim_{#1}}
\newcommand{\minim}{\mathop{\hbox{\rm minimize}}}
\newcommand{\maximize}[1]{\displaystyle\maxim_{#1}}
\newcommand{\maxim}{\mathop{\hbox{\rm maximize}}}

\def\spose#1{\hbox to 0pt{#1\hss}}

\def\text #1{\hbox{\quad#1\quad}}

\def\nthinsp{\mskip -2   mu}

\usepackage{epsfig}
\usepackage{longtable}
\usepackage{caption}
\usepackage{color}

\def\texitem#1{\par\smallskip\noindent\hangindent 25pt
               \hbox to 25pt {\hss #1 ~}\ignorespaces}

\usepackage[finalnew]{trackchanges}
\usepackage{latexsym,amsmath}

\def\bkE{{\rm I\kern-.17em E}}
\def\bk1{{\rm 1\kern-.17em l}}
\def\bkD{{\rm I\kern-.17em D}}
\def\bkR{{\rm I\kern-.17em R}}
\def\bkP{{\rm I\kern-.17em P}}

\def\bkZ{{\bf{Z}}}
\begin{document}
\title{Non-asymptotic Upper Bounds for Deletion Correcting Codes}
\author{Ankur A. Kulkarni\thanks{Both authors are at the Coordinated Science Laboratory at the University of Illinois at Urbana-Champaign, Urbana, Illinois, U.S.A., 61801. They can be reached at \texttt{akulkar3@illinois.edu} and \texttt{kiyavash@illinois.edu}, respectively. This work was supported in part by AFOSR under Grants FA9550-11-1-0016, FA9550-10-1-0573, and NSF grants CCF 10-54937 CAR and CCF 10-65022 Kiyavash.} \qquad Negar Kiyavash}
\date{}
\maketitle
\begin{abstract}
Explicit non-asymptotic upper bounds on the sizes of multiple-deletion correcting codes are presented. In particular, the largest single-deletion correcting code for $q$-ary alphabet and string length $n$ is shown to be of size at most $\frac{q^n-q}{(q-1)(n-1)}$. 
An improved bound on the asymptotic rate function is obtained as a corollary. Upper bounds are also derived on sizes of codes for a constrained source that does not necessarily comprise of all strings of a particular length, and this idea is demonstrated by application to sets of run-length limited strings.

The problem of finding the largest deletion correcting code is modeled as a matching problem on a hypergraph. This problem is formulated as an integer linear program. The upper bound is obtained by the construction of a feasible point for the dual of the linear programming relaxation of this integer linear program. 

The non-asymptotic bounds derived imply the known asymptotic bounds of Levenshtein and Tenengolts and improve on known non-asymptotic bounds. Numerical results support the conjecture that in the binary case, the Varshamov-Tenengolts codes are the largest single-deletion correcting codes.
\end{abstract}

\section{Introduction}
A \textit{deletion channel} is a communication channel that takes a string of symbols as its input and transmits only a subset of the input symbols leaving the order of the symbols unchanged. Symbols that are not transmitted constitute the errors in the channel and are called \textit{deletions}. A deletion channel is distinct from the widely studied erasure channel wherein  the positions of the errors are \change{fixed}{known}. This paper mainly concerns deletion channels where the \change{deletion probability is unity. Of these we are concerned with channels where at most one symbol is deleted, called \textit{single-deletion channels}.}{maximum number of deletions, denoted $s$, is fixed.} 

A \textit{codebook} or a \textit{deletion correcting code} for the deletion channel is a set $C$ of input strings, no two of which on transmission through the channel can result in the same output. 
For a string $x$, call the set of strings obtained by deletion of  \add{$s$} symbols from $x$, the $s$-\textit{deletion set} of $x$. An $s$-deletion correcting code  is thus a set of input strings with pairwise disjoint $s$-deletions sets. 

To explain our contribution, consider the case where $s=1$ (the \textit{single}-deletion channel). An open problem pertaining to this channel is the determination of the size of the largest or \textit{optimal} codebook $C=\Cscr_n^*$, for input strings comprising of all strings of length $n$~\cite{sloane02single}. The classical bound of Levenshtein~\cite{levenshtein66binary} provides one benchmark for optimality. For the case of binary strings, Levenshtein \cite{levenshtein66binary} showed that the size $|\Cscr^*_{n}|$ of an optimal codebook for the single-deletion channel is \textit{asymptotically} at most $\frac{2^n}{n}$. It is important to note here the sense in which this asymptoticity is being defined.  A function $f: \Nbb \rightarrow \Real$ is said to be asymptotically less than or equal to another function $g: \Nbb \rightarrow \Real$, written $f\lesssim g$, if $\lim_{n \rightarrow \infty} \frac{f(n)}{g(n)} \leq 1$. $f$ is said to be asymptotically equal to $g$, written $f\sim g$, if $f\lesssim g$ and $g\lesssim f$. Thus Levenshtein's result says that $\lim_{n\rightarrow \infty} \frac{|\Cscr^*_n|}{2^n/n} \leq 1$. Levenshtein then constructs a codebook of size at least $\frac{2^n}{n+1}$, thereby proving $\frac{2^n}{n} \lesssim |\Cscr^*_n|$, and hence concludes that the optimal codebook $\Cscr^*_n$ has size  asymptotically equal to $\frac{2^n}{n}$, \ie $\Cscr^*_n$ satisfies $\lim_{n \rightarrow \infty}\frac{|\Cscr^*_n|}{2^n/n}=1$.

If the function $g$ is bounded, the  asymptotic equality $f\sim g$ implies equality of the limiting values of $f(n)$ and $g(n)$ or their near-equality for sufficiently large  $n$. However since $g(n)=2^n/n$ is unbounded, Levenshtein's asymptotic results do not allow one to obtain a fine approximation to $|\Cscr^*_n|$, or conclude if for a particular $n$, $|\Cscr^*_n|$ is greater or less than $\frac{2^n}{n}$, or even conclude the boundedness or unboundedness of the difference $||\Cscr^*_n| -\frac{2^n}{n}|$. Indeed, the best known codes for the binary version of this channel, the Varshamov-Tenengolts (VT) codes~\cite{varshamov65codes}, are  of size at least $\frac{2^n}{n+1}$ for input length $n$. Although this sequence is asymptotically equal  to $\frac{2^n}{n}$ (and recently verified by exact search to be optimal for string lengths $n \leq 10$~\cite{sloaneweb}), the difference $\frac{2^n}{n}-\frac{2^n}{n+1}$ grows to infinity. 
%Besides, asymptotically, the sequences $\frac{2^n}{n}, \frac{2^n}{n+1}$, $\frac{2^n}{n-1}$ and $\frac{2^n}{n+o(n)}$ are all equal, which further complicates the matter of deciding the optimality of a code construction.

In other words, for this problem, asymptotic optimality of a codebook does not say much about its optimality per se. The challenges noted above continue to hold (and are perhaps more severe) for larger alphabet and larger number of deletions. For the case of multiple deletions, asymptotic bounds exist, thanks to Levenshtein~\cite{levenshtein66binary} for binary alphabet, but little is known about the quality of these bounds, since no matching lower bounds exist. A more useful bound for any such channel would be  a \textit{non-asymptotic upper bound} that also implies known asymptotic bounds.  Such a bound can serve as a hard bound on the size of a codebook for any string length and help in assessing the quality of specific code constructions.  Such non-asymptotic upper bounds are the subject of this paper. 

%The methodology for obtaining these bounds is not limited to the case of a single deletion. 
We derive explicit non-asymptotic upper bounds on the sizes of codebooks for any number of deletions $s$ and any alphabet size $q$. These bounds imply the known asymptotic bounds of Levenshtein~\cite{levenshtein66binary} and generalize them to larger alphabet. For the case of a single deletion we obtain this bound in closed form. We show that for string length $n$, an optimal $q$-ary single-deletion codebook has size at most $\frac{q^n-q}{(q-1)(n-1)}$. This implies the asymptotic upper bound of $\frac{q^n}{(q-1)n}$ shown by Tenengolts \cite{tenengolts84nonbinary}. 
% We show that the largest binary single-deletion correcting code for strings of length $n$ is of size at most . 
In the binary case, together with the size of the VT codes (which effectively provide non-asymptotic \textit{lower} bounds), our upper bound $\frac{2^n-2}{n-1}$ implies Levenshtein's asymptotic results. 

From these bounds we derive an upper bound on the \textit{asymptotic rate function}. For a channel where the number of deletions is a constant fraction of string length, this function gives the asymptotic value of the rate of the largest deletion correcting code, as a function of the fraction of symbols that are deleted. This bound on the rate function improves on the previous bound shown by Levenshtein~\cite{levenshtein02bounds}.

 We then extend this methodology to derive bounds on deletion correcting codes for \textit{constrained sources}. These are codebooks for a specific set of strings, \ie, not necessarily the set of \textit{all} strings of a particular length. Recording systems such as magnetic tapes impose physical constraints on the patterns that symbols can take in codewords~\cite{huffman98handbook}. If such a code is subsequently transmitted through a deletion channel, the codewords can be thought of as a constrained source. As a specific demonstration of this idea, we derive non-asymptotic upper bounds on sizes of codebooks for run-length limited sources for the single-deletion channel.

The bounds are obtained as follows. We characterize the largest codebook for the \remove{single-}deletion channel as a maximum \textit{matching} on a suitably defined hypergraph. The problem of finding a maximum matching is written as a 0-1 integer linear program. The \textit{fractional matching} on this hypergraph is the solution of the linear programming relaxation of this integer linear program, and its value is an upper bound on the size of the maximum matching. Our upper bound is obtained by constructing a feasible solution for the \textit{dual} of this linear program. \add{For the single-deletion channel} the construction is such that it allows for the calculation of the dual objective in closed form as $\frac{q^n-q}{(q-1)(n-1)}$. Unfortunately, for larger number of deletions, due to the complicated nature of the resulting expressions, we are unable to produce closed form expressions.

Computations on a computer reveal that for the binary single-deletion channel the optimal fractional matching size is quite close to the size of the  VT codes. For strings of length up to 14, the difference between the size of the VT codes and the optimal fractional matching is at most $8$; this indicates that the VT codes are either optimal or very close to being optimal (at least up to string length 14). On a side note, the hypergraph approach also appears to be more amenable to algorithmic approaches due to its compact representation; this aspect of this paper may be of independent interest. 

\subsection{Related work}
\change{An}{A} wide-ranging survey on various results and challenges associated with deletion correction and its variants was recently presented by Mercier \etal \cite{mercier10survey}. Sloane's survey~\cite{sloane02single} deals specifically with the binary single-deletion channel and illuminates several deep open questions pertaining to the VT codes. Here we recall some highlights from this area of work. 

The study of the deletion channel has a long history going back at least to the seminal work of Levenshtein~\cite{levenshtein66binary} wherein asymptotic bounds on the sizes of optimal binary codebooks were derived. For $s$ deletions and binary input strings,  Levenshtein~\cite{levenshtein66binary} showed that the largest codebook $\Cscr^*_{2,s,n}$ for string length $n$ satisfies the asymptotic relations
\begin{equation}
 \frac{2^s(s!)^22^n}{n^{2s}} \lesssim |\Cscr^*_{2,s,n}| \lesssim \frac{s!2^n}{n^s}. \label{eq:levenasymp}
\end{equation} 
 Levenshtein~\cite{levenshtein66binary} also noticed that the Varshamov-Tenengolts codes~\cite{varshamov65codes}, which were proposed for asymmetric error correction, served as asymptotically optimal codes for the binary single-deletion channel; these remain to date the best known codes and have recently been confirmed to be optimal for string length up to 10. An independent line of study on this topic appears to have been contemporaneously pursued by Ullman~\cite{ullman67capabilities}, \cite{ullman_near-optimal_1966}. 
 
Thereafter there have been many efforts at code construction.  An attempt at generalizing the VT codes for the binary multiple-deletion channel was made by \change{Heldberg}{Helberg} and Ferreira~\cite{helberg02multiple}; \add{that this generalization indeed corrects deletion errors was recently shown by Abdel-Ghaffar} \etal~\cite{abdel-ghaffar12helbergs}. For non-binary alphabet this problem was first studied by Calabi and Harnett~\cite{calabi69general} and Tanaka and Kasai~\cite{tanaka76synchronization}. Later Tenengolts proposed a construction similar to the VT codes for the $q$-ary single-deletion channel and showed that the optimal codebook for string length $n$, $\Cscr^*_{q,1,n}$, \change{was}{is} of size at least $\frac{q^n}{qn}$ and \change{satisfied}{satisfies} the asymptotic upper bound $|\Cscr_{q,1,n}^*| \lesssim \frac{q^n}{(q-1)n}$~\cite{tenengolts84nonbinary}. Interestingly, no asymptotic bounds for $q$-ary $s$-deletion correcting codes appear to have been explicitly articulated, though Levenshtein's original proof from~\cite{levenshtein66binary} seems extendable to $q$-ary strings. The VT codes are number-theoretic and the underlying number-theoretic logic was generalized to correct larger number of asymmetric errors by Varshamov~\cite{varshamov73class}. 

Butenko \etal attempted to find codes algorithmically by casting this problem as a maximum independent set problem on a class of graphs~\cite{butenko02finding}. Schulman and Zuckerman  considered a construction that is in part algorithmic and showed the existence of `asymptotically good' codes for deletions whose number increases proportionally to the length of the string~\cite{schulman99asymptotically}. More recently, the algorithmic approach  has been pursued by Khajouei \etal~\cite{khajouei11algorithmic} and a graph coloring based approach was studied by Cullina \etal~\cite{cullina12coloring}. Finding codes for the deletion channel, either algorithmically or through a number-theoretic construction, is a considerable challenge, as evidenced by the attempts at achieving the records for largest codebooks on the webpage maintained by Sloane~\cite{sloaneweb}.

Deletion errors have also been studied for run-length limited sources -- which we consider in this paper as a example of a constrained source -- by Roth and Siegel~\cite{roth94lee}, Hilden \etal~\cite{hilden91shift} and Bours~\cite{bours94construction}, amongst others.  However in these works, the deletion errors considered have a specific pattern and do not exactly correspond to the deletion channel we consider. Exceptions to this are the recent works of Cheng \etal \cite{cheng12moment} and Palun\v{c}i\'{c} \etal~\cite{paluncic12multiple} which consider codes for run-length limited sources for the deletion channel in its full generality.

The topic of deletion errors has spawned research on related questions, such as the existence of `perfect codes' (Levenshtein~\cite{levenshtein92perfect}), and the combinatorial problems of counting subsequences (\eg, Hirschberg and Regnier~\cite{hirschberg00tight}, Swart and Ferreira~\cite{swart03note}, Mercier \etal~\cite{mercier08number} and more recently, Liron and Langberg~\cite{liron12characterization}) and the reconstruction of sequences (Levenshtein~\cite{levenshtein01efficient,levenshtein01efficientjcomb}). Another body of active ongoing research studies the capacity of the deletion channel (\eg, Mitzenmacher~\cite{mitzenmacher06polynomial}, Kanoria and Montanari~\cite{kanoria09deletion}, and Diggavi \etal~\cite{diggavi07capacity}).

The question of non-asymptotic upper bounds, which is our interest, is comparatively less studied. One may scan Levenshtein's proof of the asymptotic bound from \cite{levenshtein66binary} to see if a non-asymptotic bound has been found in it as an intermediate step. For the single deletion channel, the bound so discovered (see Sloane's proof~\cite[Theorem 2.5]{sloane02single}) is greater than $\frac{2^n}{n-2\sqrt{n \log n}}$ (for binary alphabet) which is clearly weaker than our bound. In fact, Levenshtein~\cite{levenshtein02bounds} has presented a somewhat more general bound on the size of a $q$-ary $s$-deletion correcting code: 
\begin{equation}
|\Cscr^*_{q,s,n}| \leq \frac{q^{n-s}}{\sum_{i=0}^s\binom{r-s+1}{i} } + q \sum_{i=0}^{r-1}\binom{n-1}{i}(q-1)^i, \label{eq:leven}
\end{equation}
where $r$ is any integer satisfying $1\leq s\leq r+1 \leq n$. It is not clear which value of $r$ provides the strongest bound of these  (although a heuristic argument using Stirling's approximation suggests that $r \approx \frac{n}{2}$ should be optimal in the binary \add{single-deletion} case; this is essentially Levenshtein's original argument~\cite{levenshtein66binary}). We have found via numerical calculation that the strongest of the bounds in \eqref{eq:leven} is weaker than our bound. Additionally, our bound \add{in the single-deletion case} also has the attractiveness of being in closed form. 
Levenshtein in another paper derives another non-asymptotic bound for the size of a $q$-ary single-deletion codebook~\cite[Theorem 5.1]{levenshtein92perfect},
\begin{equation}
|\Cscr_{q,1,n}^*| \leq \frac{q^{n-1}+(n-2)q^{n-2}+q}{n}, \label{eq:leven2}
\end{equation}
but this bound is asymptotically much weaker than Tenengolts' asymptotic bound of $\frac{q^n}{(q-1)n}$ (their ratio grows to infinity; \add{our bound implies Tenengolts' asymptotic bound}). Sloane's website~\cite{sloaneweb} contains several numerical bounds found by calculating the Lov\'{a}sz $\vartheta$~\cite{west00introduction} on certain graphs. But unlike our bounds, there are no expressions (closed form or otherwise) for these bounds.

The scarcity of non-asymptotic upper bounds is perhaps due to the property that deletion sets of distinct strings can have distinct sizes. This point has also been stressed by Sloane~\cite[Section ``Optimality'']{sloane02single}: \textit{``It is more difficult to obtain upper bounds for deletion-correcting codes than for conventional error-correcting codes, since the disjoint balls $D_e(u)$ (deletion sets) associated with the codewords ... do not all have the same size. Furthermore the metric space ($\Fbb^n_2, d$)\footnote{$d$ is the Levenshtein or edit distance, cf. Definition \ref{def:edit}.} is not an association scheme and so there is no obvious linear programming bound.''} In the light of this comment it is interesting that our non-asymptotic bound is obtained from a linear programming argument, and it relies critically on the sizes of the deletion sets.

\subsection{Organization}
This paper is organized as follows. Section \ref{sec:pre} comprises of preliminaries including, notation, problem definition, background on hypergraphs and the derivation of lemmas that are of use in our analysis. Section \ref{sec:bounds} contains the hypergraph characterization of the optimal codebook and the derivation of the upper bounds for single-deletion correcting codes. In Section \ref{sec:larger} we extend the analysis to obtain bounds on codes for larger number of deletions and derive a bound on the asymptotic rate function. 
In Section \ref{sec:constrained}, we derive bounds on codebooks for constrained sources, in particular, for run-length limited sources. Numerical simulations comparing the values of Levenshtein's bound from \eqref{eq:leven}, our bound, the tightest bound obtainable by our logic,  and the best known codes are presented in Section \ref{sec:numeric}. In Section \ref{sec:disc} we discuss our results and possible avenues for tightening our bound and  conclude the paper.

\section{Preliminaries} \label{sec:pre}
Let $\Fbb_q = \{0,1,\hdots, q-1\}$ be a $q$-ary alphabet and let $\Fbb_q^n$  denote the set of all $q$-ary sequences of length $n$. Any such $q$-ary sequence is called a \textit{string}. We let $\Fbb_q^* = \bigcup_{n=0}^\infty \Fbb_q^n$ denote set of all strings; here $\Fbb_q^0$ denotes the empty string. 
Let $x= x_1\hdots x_n$ be a string. A \textit{subsequence} of $x$ is formed by taking a subset of the symbols of 
$x$ and aligning them without altering their order. In other words, a {subsequence} of $x$ is a sequence $y=x_{i_1}\hdots x_{i_k}$, where $1 \leq k \leq n$  and the indices satisfy $1 \leq i_1 < \hdots < i_k \leq n$;  $x$ is called a \textit{supersequence} of $y$. We say that $y$ is obtained from $x$ by the \textit{deletion} of $n-k$ symbols and $x$ is 
obtained from $y$ by the \textit{insertion} of $n-k$ symbols. 
%For any $x,y \in \Fbb^*_2$, let $l(x),l(y)$ be the lengths of $x$ and $y$, let $\bar{l}(x,y)$ be the minimum length of a common supersequence of $x,y$ and let $\underline{l}(x,y)$ be maximum length of a common subsequence of $x,y$. The empty strings is a common subsequence of any two strings and concatenation of the strings is a common supersequence, whereby these $\bar{l}$ and $\underline{l}$ are well-defined and finite on $\Fbb_2^*\times \Fbb_2^*$. The Levenshtein distance between $x$ and $y$
%has the following characterization \cite{levenshtein92perfect}
%\begin{equation}
% d(x,y) = l(x) + l(y) -2\underline{l}(x,y) = 2\bar{l}(x,y) -  l(x) - l(y) = \bar{l}(x,y) - \underline{l}(x,y). \label{eq:leven}
%\end{equation} 
% It is easy to show that $d$ is a metric. 

A specific type of subsequence that is important for our results is a \textit{run}, defined below.
\begin{definition} \label{def:run}
Let $x=x_1\hdots x_n \in \Fbb_q^n$ be a string. A \textit{run} of $x$ is a maximal contiguous subsequence with identical symbols, \ie a run of $x$ is a sequence  $x_ix_{i+1}\hdots x_{i+j}$, $1 \leq i \leq i+j \leq n$ with the property that $x_i=x_{i+1} =\hdots =x_{i+j}$ and the properties that, \add{a)} if \change{$1<i\leq i+j <n$, $x_{i-1} \neq x_i$ and $ x_{i+j} \neq x_{i+j+1}$}{$1<i$ then $x_{i-1} \neq x_i$, and b) if $i+j <n$, then $ x_{i+j} \neq x_{i+j+1}$}. For any $x \in \Fbb_q^*$, $r(x)$ denotes the number of runs of $x$.
\end{definition}
For example if $q=3$ and $x =120010$, the runs of $x$ are $1, 2, 00, 1, 0$ and $r(x)=5$. Clearly for any $x \in \Fbb_q^n$, $1\leq r(x) \leq n$.

\begin{definition}
For any string $x \in\Fbb^{*}_q$, the set of subsequences \add{of} $x$ obtained by deletion of $s$ symbols is denoted by $D_s(x)$ and set of supersequences obtained by insertion of $s$ symbols into $x$ is denoted by $I_{s}(x)$. We call $D_s(x)$ and $I_s(x)$ the $s$-\textit{deletion set} of $x$ and $s$-\textit{insertion set} of $x$, respectively.  
\end{definition}
For example if $q=3, s=1$ and $x=120010$, then $D_1(x) = \{20010, 10010, 12010, 12000, 12001\}$. Notice that subsequences obtained by the 
deletion of a symbol from the same run of $x$ are all identical. For example,  in the run $00$, deletion of either $0$ results in the same subsequence $12010$. 
Consequently we have the following relation~\cite{levenshtein92perfect}, 
\begin{equation}
|D_1(x)| = r(x), \qquad \forall x \in \Fbb_q^*. \label{eq:delsize}
\end{equation}
For $s>1$, expressions for $|D_s(x)|$ get increasingly complicated, and depend on statistics of $x$ other than the number of runs (see, \eg,~\cite{mercier08number} for one set of expressions). We discuss bounds on $|D_s(\cdot)|$ later in Section \ref{sec:larger}.

Surprisingly, the size of $I_s(x)$ is independent of $x$, but is a function only of the length of $x$ and the size of the alphabet~\cite[Lemma 1, p.\, 354]{sankoff83time}. \note{need to fix this ref acc Navin's suggestion} Specifically, we have
\begin{eqnarray}
|I_{s}(x)|  &=& \sum_{j=0}^s\binom{n}{j}(q-1)^j  \qquad \forall\ x \in \Fbb_q^{n-s}.\label{eq:inssize} 
\end{eqnarray}
We denote this quantity by $\iota_{q,s,n}$,
\begin{equation}
\iota_{q,s,n} \triangleq \sum_{j=0}^s\binom{n}{j}(q-1)^j. \label{eq:iqn}
\end{equation}
As a general rule, instead of using `$1$-deletion' or `$1$-insertion' (correcting code, set,$\hdots$), we use the more elegant `\textit{single}-deletion' (correcting code, set, $\hdots$) etc.

The central object of our interest, namely, a deletion correcting code is defined below.
\begin{definition} \label{def:delcode}
A $s$-deletion correcting code (or ``$s$-deletion codebook'') for string length $n$ and alphabet $\Fbb_q$ is a set $C \subseteq \Fbb_q^n$ with the property 
that the sets $D_s(x), x\in C$, are pairwise disjoint. The largest such code is denoted by $\Cscr^*_{q,s,n}$ and called an optimal $s$-deletion correcting code or optimal $s$-deletion codebook. 
\end{definition}

A code capable of correcting $s$ deletions is also capable of correcting a total of $s$ insertions and deletions~\cite{levenshtein66binary}, whereby an $s$-deletion correcting code is also a $s$-insertion correcting code (\ie, a set $C\subseteq \Fbb_q^n$ such that the sets $I_s(x), x\in C$, are pairwise disjoint)~\cite{levenshtein66binary}. Another characterization of single-deletion correcting codes is through 
the Levenshtein distance. 
\begin{definition} \label{def:edit}
For any $x,y \in \Fbb_q^*$ define the Levenshtein distance or edit distance $d(x,y)$ as minimum number of insertions 
or deletions required to obtain $x$ from $y$.  
\end{definition}
A set $C\subseteq \Fbb_q^n$ is a $s$-deletion correcting code  if and only if $d(x,y)>2s$ for any two distinct 
strings $x,y\in C$. In summary, we have the following equivalence~\cite{levenshtein66binary}.
\begin{lemma} \label{lem:insd}
For any $x,y \in \Fbb_q^n$, the following three statements are equivalent. 
\begin{enumerate}
 \item $d(x,y) \leq 2s,$ 
\item $D_s(x) \cap D_s(y) \neq \emptyset,$
\item $I_s(x) \cap I_s(y) \neq \emptyset.$
\end{enumerate}
\end{lemma}

The following lemma, although not directly related to deletion correction, will be required for our analysis.
\begin{lemma} \label{lem:eqsol} Let $n,k,d \in \Nbb, k\leq n, dk \leq n$ and let $t_1,\hdots,t_k$ be variables taking values in $\Nbb$. The number of solutions 
$(t_1,\hdots,t_k)$ to the set of equations 
\begin{equation}
\sum_{i=1}^k t_i = n, \qquad t_i \geq d, t_i \in \Nbb, \forall \ 1 \leq i \leq n, \label{eq:sys}
\end{equation}
is $\binom{n-k(d-1)-1}{k-1}.$
\end{lemma}
\begin{proof}
First suppose $d=1$. Consider an array of $n$ $1$'s and insert $k-1$ $0$'s between the $1$'s, so that no two $0$'s are inserted next to each other and no $0$'s are inserted at the beginning  or the end of the array. There is a one-to-one correspondence between an arrangement of this kind and a solution of \eqref{eq:sys}: $t_i$, for $1<i<k$, corresponds to the number of 1's between the $(i-1)\th$ $0$ and $i\th$ $0$ and $t_1,t_k$ are the number of 1's at the beginning and the end of the array. The number of such arrangements is easily seen to be $\binom{n-1}{k-1}$.

Now suppose $d>1$. Notice that the system \eqref{eq:sys} is equivalent to the system
\begin{align*}
\sum_{i=1}^k (t_i -(d-1))&= n-k(d-1), \\ 
 (t_i-(d-1)) &\geq 1, t_i -(d-1)\in \Nbb, \forall \ 1 \leq i \leq n. 
\end{align*}
This system reduces to the earlier case with $d=1$, but with variables $t_i' = t_i-(d-1)$, for $i=1,\hdots,k$. The number of solutions in this case is $\binom{n-k(d-1)-1}{k-1}.$
\end{proof}
\subsection{Background on hypergraphs}\label{sec:back}
The contents of this section are sourced from Berge~\cite{berge89hypergraphs}. 

A hypergraph is a generalization of the concept of a graph. In a graph edges are pairs of vertices. In a hypergraph, one allows arbitrary nonempty sets of vertices, including those with exactly one element, to be the \change{so called}{so-called} \textit{hyperedges}. Formally, 
\begin{definition}
A hypergraph $\Hscr$ is a tuple $(X,\Escr)$, where $X$ is a finite set and $\Escr$ is a collection of nonempty subsets of $X$ such that $\bigcup_{E \in \Escr} E=X$. 
$X$ is called the \textit{vertex set}, its elements are called \textit{vertices} and the elements of $\Escr$ are called \textit{hyperedges}. 
\end{definition}
When a vertex belongs to a hyperedge, we say it is \textit{covered} by the hyperedge. The above definition assumes that the hypergraph contains no exposed  vertex, \ie, a vertex that is covered by no hyperedge. This is a matter of convention; other definitions, \eg~\cite{scheinerman01fractional},  do not impose this requirement. 

Let $\Escr = \{E_1,\hdots,E_m\}$ be the set of hyperedges of the hypergraph $\Hscr=(X,\Escr)$. For a set of indices $J \subseteq \{1,\hdots,m\}$, the  \textit{partial hypergraph generated by $J$} is $\Hscr_J =(X_J, \{E_j | j \in J\}),$ where $X_J=\bigcup_{j\in J} E_j$. 
%For a set $A \subseteq X,$ the hypergraph 
%$\Hscr_{|A} = (A,\{E \cap A | E \in \Escr\})$ is called the \textit{sub-hypergraph induced by $A$}.  

Hyperedges are defined as sets and as such one can talk of intersection of hyperedges. Specifically, two hyperedges are disjoint if there is no vertex that is covered by both hyperedges. The idea of packing neighborhoods or spheres used in coding theory sits naturally in the theory of hypergraphs. A packing of hyperedges is called a matching.
\begin{definition}
A matching of a hypergraph $\Hscr=(X,\Escr)$ is a collection of pairwise disjoint hyperedges $E_1,\hdots, E_j \in \Escr$. The matching number of  $\Hscr$, denoted $\nu(\Hscr)$, is the largest $j$ for which such a matching exists. 
\end{definition}
A dual concept (in a sense we make precise below) of a matching is a transversal.
\begin{definition}
A transversal of a hypergraph $\Hscr=(X,\Escr)$ is a subset $T \subset X$ that intersects every hyperedge in $\Escr$. The transversal number of $\Hscr$, denoted $\tau(\Hscr)$, is the smallest size of a transversal. 
\end{definition}

%\begin{definition}
%Let $\Hscr=(X,\Escr)$ be a hypergraph where $X=\{x_1,\hdots,x_{|X|}\}$ and $\Escr = \{E_1,\hdots,E_{|\Escr|}\}$. The dual of $\Hscr$, denoted $\Hscr^*$, is a hypergraph whose vertices $e_1,\hdots,e_{|\Escr|}$ are the hyperedges of $\Hscr$ and with hyperedges 
%\[  Y_i = \{ e_j | x_i \in E_j\ {\rm in} \ \Hscr \}  \]
%\end{definition}
%
%\begin{lemma}
%$\nu(\Hscr^*) = p(\Hscr)$ and $\tau(\Hscr^*) = \kappa(\Hscr)$. 
%\end{lemma}

Suppose $\Hscr=(X,\Escr)$ is a hypergraph with $n$ vertices $x_1,\hdots,x_{n}$ and $m$ hyperedges $E_1,\hdots,E_{m}$. Consider a matrix 
$A \in \{0,1\}^{n \times m}$, where the element in the $i\th$ row and $j\th$ column is 
\[A[i,j] = \begin{cases}
               1 & \eef x_i \in E_j,  \\
0                & \ow.
              \end{cases}
    \]
$A$ is called the incidence matrix of $\Hscr$. The matching number and the transversal number are both solutions of integer linear programs. 
In the rest of this paper, we refer to problem \eqref{eq:matprob} below as the \textit{matching problem} and \eqref{eq:tranprob} as the \textit{transversal problem} on hypergraph $\Hscr$. 
\begin{lemma}
The matching number and transversal number are solutions of integer linear programs: 
\begin{align}
 \nu(\Hscr)  = \max \{ \bfone\t z  |\ Az \leq \bfone, z_j \in  \{0,1\}, 1\leq j\leq m\}, \label{eq:matprob}\\
\tau(\Hscr) = \min \{ \bfone\t w | A\t w \geq \bfone, w_i \in  \{0,1\}, 1\leq i \leq n\}, \label{eq:tranprob}
\end{align}
where $\bfone$ denotes a column vector of all $1$'s of appropriate dimension.
\end{lemma}
\begin{proof}
In the integer linear programming formulation of the matching problem, each hyperedge $E_j \in \Escr$ corresponds to a variable $z_j \in \{0,1\}$ and $z$ is the vector $(z_1,\hdots,z_m)$. The variable $z_j$ is interpreted as the indicator function that identifies if hyperedge $E_j$ is a part of the matching represented by $z$. Thus $z_j=1$ if $E_j$ is selected, and $z_j=0$ otherwise.  The matching problem has one constraint for each vertex: for a vertex $x_i$, the sum of $z_j$ over those hyperedges $j$ that cover vertex $x_i$ is at most $1$; hence, at most one of  these $z_j$ takes value $1$. Consequently, a vector $z$ is feasible for the matching problem if and only if the collection $\{E_j:z_j =1\}$ is a matching of $\Hscr$. It follows that the matching number of $\Hscr$ is the optimal value of \eqref{eq:matprob}. 

By a similar construction, in the integer linear programming formulation of the transversal problem, let each vertex $x_i\in X$ correspond to a variable $w_i\in \{0,1\}$ and let $w=(w_1,\hdots,w_n)$. The variable $w_i=1$ if and only if vertex $x_i$ is included in the transversal represented by $w$. The transversal problem has one constraint for each hyperedge which says that for a hyperedge $E_j$, the sum of $w_i$  over those vertices $i$ that are covered by $E_j$ is at least $1$, whereby at least one \add{of} these $w_i$ \remove{is} takes value $1$. There is thus a one-to-one correspondence between a transversal of $\Hscr$ and a feasible vector $w$ for \eqref{eq:tranprob}. The transversal number is thus characterized by \eqref{eq:tranprob}.
\end{proof}

Notice that the mathematical programs in \eqref{eq:matprob} and \eqref{eq:tranprob} are duals of each other. A fundamental theorem of integer linear programming states that a pair of dual programs satisfy \textit{weak duality}. Weak duality means that of the pair of dual problems, the value of the maximization problem is no greater than the value of the minimization problem~\cite{schrijver98theory}. Applied to \eqref{eq:matprob}-\eqref{eq:tranprob}, this implies, for any hypergraph $\Hscr$,
\begin{equation}
\nu(\Hscr) \leq \tau(\Hscr). \label{eq:weak}
\end{equation}

We note a technical point about problems \eqref{eq:matprob}-\eqref{eq:tranprob} that helps in simplifying our analysis. Notice that the constraint $z_j \in\{0, 1\}$ in \eqref{eq:matprob} and the constraint $w_i \in \{0,1\}$ in \eqref{eq:tranprob} may as well be replaced with the constraints $z_j \in \Zbb_+$ and $w_i \in \Zbb_+$, respectively, where $\Zbb_+$ is the set of nonnegative integers, to give the following equivalent characterizations for $\nu(\Hscr)$ and $\tau(\Hscr)$
\begin{align}
 \nu(\Hscr)  &= \max \{ \bfone\t z  |\ Az \leq \bfone, z_j \in  \Zbb_+, 1\leq j\leq m\}, \label{eq:matprob2}\\
\tau(\Hscr) &= \min \{ \bfone\t w | A\t w \geq \bfone, w_i \in  \Zbb_+, 1\leq i \leq n\}. \label{eq:tranprob2}
\end{align}
To see the equivalence between \eqref{eq:matprob} and \eqref{eq:matprob2}, notice that no vector $z \in \Zbb_+^m$ satisfying $Az \leq \bfone$ can have a component greater than $1$. And in \eqref{eq:tranprob}, observe that no minimizing $w \in \Zbb_+^n$ of \eqref{eq:tranprob2} can have 
a component greater than $1$. From now on, we consider only the formulations \eqref{eq:matprob2}-\eqref{eq:tranprob2}. Note that sources such as Berge~\cite{berge89hypergraphs} omit the above analysis and directly employ \eqref{eq:matprob2}-\eqref{eq:tranprob2} to define $\nu(\Hscr)$ and $\tau(\Hscr)$.

The linear programming relaxation of an integer program is constructed by replacing the requirement that a variable \change{take}{takes} only integral values  by a requirement that allows the variable to also take any real value between the integral values (\ie, in the convex hull of the integral values)~\cite{schrijver98theory}. 
By $\nu^*(\Hscr)$ and $\tau^*(\Hscr)$ we denote the values of the linear programming relaxations of \eqref{eq:matprob2} and \eqref{eq:tranprob2}, respectively. \ie,  
\begin{align}
 \nu^*(\Hscr)  &= \max \{ \bfone\t z  |\ Az \leq \bfone, z \geq 0\}, \label{eq:frmatprob}\\
\tau^*(\Hscr) &= \min \{ \bfone\t w | A\t w \geq \bfone, w\geq 0 \}, \label{eq:frtranprob}
\end{align}
where for simplicity, we denote a vector of zeros of appropriate size also by `$0$'. 
$\nu^*(\Hscr)$ and $\tau^*(\Hscr)$ are called the \textit{fractional matching number} and \textit{fractional transversal number} of $\Hscr$.
A vector $z$ feasible for \eqref{eq:frmatprob} is called a \textit{fractional matching} and the set $\{z: Az \leq \bfone, z \geq 0\}$ is called the \textit{fractional matching polytope} of $\Hscr$. A vector $w$ feasible for \eqref{eq:frtranprob} is called a \textit{fractional transversal} and the set $\{w: A\t w \geq \bfone, w \geq 0\}$ is called the \textit{fractional transversal polytope}. $\bfone\t z$ and $\bfone \t w$ are called the \textit{weights} of $z$ and $w$. $\nu^*(\Hscr)$ and $\tau^*(\Hscr)$ being linear programs satisfy the fundamental property of  \textit{strong duality}~\cite{schrijver98theory}, \ie,  
\[\nu^*(\Hscr) =\tau^*(\Hscr).\] 
Thus for any hypergraph the fractional matching number and the fractional transversal number are equal. In general, integer programs do not satisfy strong duality and thereby equality may not hold in \eqref{eq:weak}. Equality or lack thereof in \eqref{eq:weak} depends on the shape of the fractional matching and fractional transversal polytopes.  On a side note, we recall that linear programming relaxations have been employed in the decoding of binary linear codes by Feldman \etal~\cite{feldman05using}.

Fractional matchings and transversals do not have as direct a counting interpretation as the vectors feasible for \eqref{eq:matprob}-\eqref{eq:tranprob}. However they are extremely useful for obtaining bounds. Since the feasible regions of the integer programs are strictly contained in the feasible regions of their of the linear programming relaxations, we immediately have $\nu(\Hscr) \leq \nu^*(\Hscr)$ and $\tau^*(\Hscr) \leq \tau(\Hscr)$. Furthermore, we have the following lemma.
\begin{lemma} \label{lem:duality}
For any hypergraph $\Hscr$, we have 
\begin{align*}
 \nu(\Hscr) \leq \nu^*(\Hscr) = \tau^*(\Hscr) \leq \tau(\Hscr).
\end{align*}
In particular, 
\[\nu(\Hscr) \leq \tau^*(\Hscr) \leq \bfone \t w,\]
for any fractional transversal $w$. 
\end{lemma}
\begin{proof}
Since fractional matchings and transversal problems are relaxations of the matching and transversal problem, $\nu(\Hscr) \leq\nu^*(\Hscr)$ and 
$\tau^*(\Hscr) \leq \tau(\Hscr)$. By the duality theorem of linear programming $\nu^*(\Hscr)$ and $\tau^*(\Hscr)$ are equal. By definition, any fractional transversal $w$ must have weight no less than the fractional transversal number, by which the last claim follows.
\end{proof}

We end this survey with one final concept, that of a line graph. 
\begin{definition}
A line graph of a hypergraph $\Hscr=(X,\Escr)$ is a graph $L(\Hscr)$ with vertices given by the hyperedges of $\Hscr$ and 
two vertices in $L(\Hscr)$ are joined by an edge if they intersect as hyperedges in $\Hscr$.
\end{definition}
An independent set of a graph is a set of vertices, no two of which share an edge. For a graph $G$ we denote the size of its largest independent set, or its \textit{independence number}, by $\alpha(G)$. Now consider a hypergraph $\Hscr$. An independent set of its line graph $L(\Hscr)$ corresponds to a collection of hyperedges of $\Hscr$ that are pairwise disjoint. Consequently, 
\begin{equation}
\nu(\Hscr) = \alpha(L(\Hscr)), \label{eq:line}
\end{equation}
\ie, the matching number of a hypergraph equals the independence number of its line graph.

\section{Non-asymptotic upper bounds for single-deletion correcting codes} \label{sec:bounds}
\subsection{Hypergraph characterization} \label{sec:hyperchar}
\def\Hdqn{\Hscr^\Dsf_{q,1,n}}
\def\Hdqsn{\Hscr^\Dsf_{q,s,n}}
\def\Hcqn{\Hscr^\Csf_{s,n}}
\def\Hiqn{\Hscr^\Isf_{q,1,n}}
\def\Hiqsn{\Hscr^\Isf_{q,s,n}}
\def\Hi1n{\Hscr^\Isf_{1,n}}
\def\Hnsn{\Hscr^\Nsf_{s,n}}
\def\Adsn{A^\Dsf_{s,n}}
\def\Acsn{A^\Csf_{s,n}}
\def\Aisn{A^\Isf_{s,n}}
\def\Ai1n{A^\Isf_{1,n}}
\def\Ansn{A^\Nsf_{s,n}}

\def\Ainm{A^\Isf_{1,n-1}}
\def\Pnmneg{P^{n-1}_\neg}
\def\Pn2neg{P^{n-2}_\neg}
\def\Pnneg{P^{n}_\neg}
The contents of this subsection apply to any $s$ number of deletions. We will specialize to single-deletions and present our bounds in the following subsection. 

Consider the following hypergraphs.
\begin{align*}
\Hdqsn &= (\Fbb^{n-s}_q, \{D_s(x) | x \in \Fbb_q^{n} \}),\\
 \Hiqsn &= (\Fbb_q^{n+s}, \{I_s(x) | x \in \Fbb_q^{n} \}).
\end{align*}
In each of these hypergraphs, hyperedges correspond to strings in $\Fbb_q^n$ and the vertices are strings in $\Fbb_q^{n-s}$ and $\Fbb_q^{n+s}$ for $\Hdqsn$ and $\Hiqsn$, respectively. By Definition \ref{def:delcode}, an $s$-deletion correcting code in $\Fbb_q^n$ corresponds to disjoint hyperedges in $\Hdqsn$ and therefore corresponds to a \textit{matching} in $\Hdqsn$. The size of the largest codebook for string length $n$, $|\Cscr^*_{q,s,n}|$ is thus equal to $\nu(\Hdqsn)$, the matching number of $\Hdqsn$. The matching problem for $\Hdqsn$ when written explicitly, is as follows,
$$ \maxproblemsmall{$|\Cscr^*_{q,s,n}|=\ $}
        {z}
        {\sum_{y \in \Fbb_q^n} z(y)}
                                 {\hspace{-1cm}\begin{array}{r@{\ }c@{\ }ll}
\sum_{y \in I_s(x)} z(y) &\leq & 1, & \forall  x\in \Fbb_q^{n-s},\\
z(y)                                    &\in & \Zbb_+,  & \forall y \in \Fbb_q^{n}.                              
        \end{array}} 
$$
Here the  integer variables are denoted $z(y), y \in \Fbb_q^n$. The constraints are that for each vertex $x\in \Fbb_q^{n-s}$, the sum of $z(y)$ over those $y$ for which the hyperedge corresponding to $y$ covers $x$ (\ie, $y \in I_s(x)$) is at most unity. 
Since a code is an $s$-deletion correcting code if and only if it is an $s$-\textit{insertion} correcting code, a matching of $\Hiqsn$ also corresponds to a $s$-deletion correcting code and thereby, $\nu(\Hiqsn) = |\Cscr^*_{q,s,n}|$. 

Another characterization of the optimal codebook adopted in~\cite{cullina12coloring,khajouei11algorithmic,sloane02single} employs the following graph.
\begin{definition}
Let $L_{q,s,n}$ be the graph with vertex set $\Fbb_q^n$ wherein two vertices are adjacent if their Levenshtein distance  is at most $2s$.
\end{definition}
 The optimal $s$-deletion codebook corresponds to the maximum independent set in this graph. The Levenshtein distance (restricted to $\Fbb_q^n \times \Fbb_q^n$) is the shortest path metric on the graph $L_{q,1,n}$. The hypergraph characterization relates to this characterization through the concept of a line graph. Specifically, 
\begin{lemma} \label{lem:line} For any $q,s,n \in \Nbb$, the graph $L_{q,s,n}$ is the line graph of hypergraph $\Hdqsn$ and of hypergraph $\Hiqsn$. 
Consequently, 
\begin{align*}
\nu(\Hdqsn) &= \alpha(L_{q,s,n})=|\Cscr^*_{q,s,n}|, \\
\nu(\Hiqsn) &= \alpha(L_{q,s,n})=|\Cscr^*_{q,s,n}|.
\end{align*}
\end{lemma}
\begin{proof}
By the Definition \ref{def:edit} of Levenshtein distance and by Lemma \ref{lem:insd}, two vertices in $L_{q,s,n}$ share an edge if and only if their $s$-deletion (and $s$-insertion) sets intersect. Consequently, $L_{q,s,n} = L(\Hdqsn)=L(\Hiqsn).$ By \eqref{eq:line}, the matching numbers of $\Hdqsn$ and $\Hiqsn$ are both equal to the independence  number of $L_{q,s,n}.$ 
\end{proof}

If one attempts to upper bound the size of a code by packing graph $L_{q,s,n}$ with non-overlapping neighborhoods centered around strings in $\Fbb_q^n$, the main difficulty encountered is that the resulting neighborhoods are not of the same size. This property of the Levenshtein distance is a fundamental departure from, say, the Hamming distance under which the sizes of the neighborhoods are same for every string. 

Alternatively, one may pack $\Fbb_q^{n-s}$ with deletion sets of strings in $\Fbb_q^n$. This approach too encounters the difficulty that deletion sets are of different sizes. 
For example for $s=1$, if one argues that 
$$|\Cscr^*_{q,1,n}|\min_{x \in \Fbb_q^n}|D_1(x)| \leq \sum_{x \in \Cscr^*_{q,1,n}} |D_1(x)| \leq q^{n-1},$$  since $\min_{x\in \Fbb_q^n}|D_1(x)|=1,$ one gets the bound $|\Cscr^*_{q,1,n}| \leq q^{n-1}$ which is far weaker than the asymptotic bound (the ratio $\frac{q^{n-1}}{q^n/n(q-1)}$ approaches infinity for large $n$). A similar situation results for $s>1$. Levenshtein's bound \eqref{eq:leven} is obtained by a refinement of this approach in which strings are classified in two categories based on their number of runs.

Since insertion-correction and deletion-correction are equivalent, and since insertion sets are of the same size for each string of a given length (cf., \eqref{eq:inssize}), one may exploit this to pack $\Fbb_q^{n+s}$ with insertion sets. Unfortunately, this leads to a weak upper bound. For example, for $s=1$ we get the bound $\frac{q^{n+1}}{n(q-1)+q}$, which is asymptotically $q$ times \add{larger} than the known upper bound (this bound is $\frac{2^{n+1}}{n+1}$ for binary alphabet and the asymptotic size is $\frac{2^n}{n}$). 

The approaches of packing deletion sets or insertion sets can be conceptually unified by casting them as matching problems on hypergraphs $\Hdqsn$ and $\Hiqsn$, respectively. Since insertion sets are of the same size, hypergraph $\Hiqsn$ is \textit{uniform} \cite{berge89hypergraphs}; indeed the matching problem is well studied on uniform hypergraphs (see \eg,~\cite[Chapter 3]{berge89hypergraphs},\cite{furedi81maximum} and~\cite{aharoni96theorem}). It is a quirk of the problem of deletion-correcting codes that although the characterization of $\Cscr^*_{q,s,n}$ via $\Hiqsn$ is analytically convenient and well studied, it leads to a weak bound.

The other hypergraph $\Hdqsn$ is \textit{regular}, since all vertices in $\Hdqsn$ have the same number of hyperedges covering them~\cite{berge89hypergraphs}. Although this hypergraph does not belong to a category where the matching problem appears to be well studied,  we show in the following sections that, if appropriately tackled, it does lead to a better bound. The crux of the proof of our bound lies in tackling this hypergraph.
\subsection{The non-asymptotic upper bounds for single-deletion correcting codes}
In this section we present bounds on single-deletion correcting codes. 
The bounds we obtain are based on two concepts. The first is a monotonicity relationship between the number of runs of a string (recall Definition \ref{def:run}) under the operation of insertion. The second is the property that the size of the deletion set is also equal to the number of runs (cf. \eqref{eq:delsize}). We first note the monotonicity. 
\begin{lemma} \label{lem:runs}
Let $q,n \in \Nbb$ and let  $x \in \Fbb_q^*$ be a string. Then for any supersequence $y \in I_1(x)$, the number of runs of $x$ and $y$ satisfy $r(x) \leq r(y)$. 
\end{lemma}
This lemma is quite obvious; we omit the proof for brevity. 

% With the above lemma, we are ready to prove the bound. 
Our proof utilizes Lemma \ref{lem:duality}; for easy reference the fractional transversal problem of $\Hdqn$ is written below explicitly.
$$ \problemsmall{$\tau^*(\Hdqn)=\ $}
        {w}
        {\sum_{x \in \Fbb_q^{n-1}} w(x)}
                                 {\hspace{-1cm}\begin{array}{r@{\ }c@{\ }ll}
\sum_{x \in D_1(y)} w(x) &\geq & 1, & \forall  y\in \Fbb_q^{n},\\
w(x)                                    &\geq & 0,  & \forall x \in \Fbb_q^{n-1}.                              
        \end{array}} 
$$
Notice that the variables are $w(x), x \in \Fbb_q^{n-1}$ and the constraint is that for any $y \in \Fbb_q^n$, 
the sum of $w(x)$ over those $x$ that are covered by the hyperedge corresponding to $y$ (\ie, $x\in D_1(y)$), is at least unity. 
%\begin{theorem}\label{thm:binary}
%Let $n \in \Nbb, n\geq 2$. The optimal binary single-deletion correcting code, $\Cscr^*_{2,1,n}$ satisfies 
%$$|\Cscr^*_{2,1,n}| \leq \frac{2^n-2}{n-1}.$$
%\end{theorem}
%\begin{proof}
%\begin{align*}
%\sum_{x\in \Fbb_2^{n-1}}w(x) &= \sum_{r=1}^{n-1} 2 \binom{n-2}{r-1}.\frac{1}{r} \\
%& = 2\sum_{r=1}^{n-1} \frac{(n-2)!}{(n-r-1)!(r-1)!}.\frac{1}{r} \\ 
%&\buildrel{(c)}\over=\frac{2}{(n-1)}\sum_{r=1}^{n-1} \binom{n-1}{r} \\
%&=\frac{2(2^{n-1}-\binom{n-1}{0})}{n-1}=\frac{2^n-2}{n-1}.
%\end{align*}
%\end{proof}
\begin{theorem}\label{thm:qary}
Let $q,n \in \Nbb, q \geq 2, n\geq 2$. The optimal $q$-ary single-deletion correction code $\Cscr^*_{q,1,n}$ satisfies
$$|\Cscr^*_{q,1,n}|\leq \frac{q^n-q}{(q-1)(n-1)}.$$
\end{theorem}
 \begin{proof}
By Lemma \ref{lem:line}, the size of the largest single-deletion correcting code equals the matching number of hypergraph $\Hscr_{q,1,n}^\Dsf$, \ie, $\nu(\Hscr^\Dsf_{q,1,n}) =|\Cscr^*_{q,1,n}|$. By Lemma \ref{lem:duality}, to show the required upper bound on $\nu(\Hscr_{q,1,n}^\Dsf)$ it suffices to construct a fractional transversal of  
$\Hscr_{q,1,n}^\Dsf$ with weight equal to $\frac{q^n-q}{(q-1)(n-1)}.$ To this end, consider the fractional transversal $w$, where the component of $w$ corresponding to string $x\in \Fbb_q^{n-1}$, denoted $w(x)$, is given by
\[w(x) = \frac{1}{r(x)}, \qquad \forall \ x \in \Fbb_q^{n-1},\]
where $r(x)$ is the number of runs of $x$. 
Clearly, $w \geq 0$. To show that $w$ is indeed a fractional transversal, observe that for any $y\in \Fbb_q^n$,
\[\sum_{x\in D_1(y)} w(x) = \sum_{x\in D_1(y)}  \frac{1}{r(x)} \buildrel{(a)}\over\geq \frac{|D_1(y)|}{r(y)}\buildrel{(b)}\over=1.\]
The inequality in $(a)$ follows from monotonicity relationship claimed in Lemma \ref{lem:runs} and the equality in $(b)$ follows from the size of the deletion set, given in \eqref{eq:delsize}. It only remains to calculate the weight of this transversal. For this, note that the number of strings of length $n-1$ with exactly $r$ runs is $q(q-1)^{r-1}\times \binom{n-2}{r-1}$. This is because, we have $q$ choices for the symbol of the first run and for 
every subsequent run we have $q-1$ choices for its symbol. The number of choices for the lengths of the runs equals the number of integral solutions $(t_1,\hdots,t_r)$ to 
\[\sum_{i=1}^r t_i = n-1,\quad \quad  t_i \geq 1, 1\leq i\leq r,    \]
which, by Lemma \ref{lem:eqsol}, is $\binom{n-2}{r-1}$.  Consequently, the weight of $w$ is 
%  As in the proof of Theorem \ref{thm:binary}, we prove this bound by constructing a fractional transversal. Since Lemma \ref{lem:runs} and \eqref{eq:delsize} apply for any $q$-ary alphabet, the proof  of Theorem \ref{thm:binary} shows that, $w(x) = \frac{1}{r(x)}, x\in \Fbb_q^{n-1}$ is a fractional transversal for $\Hdqn$. We now calculate the weight of $w$, noting that the number of $q$-ary strings of length $n-1$ with $r$ runs is $q(q-1)^{r-1}\binom{n-2}{r-1}$. Therefore,
 \begin{align*}
\sum_{x\in \Fbb_q^{n-1}} w(x) &= \sum_{r=1}^{n-1}q(q-1)^{r-1}\binom{n-2}{r-1}.\frac{1}{r} \\
&=q\sum_{r=1}^{n-1}\frac{(n-2)!}{(n-r-1)!(r-1)!}.\frac{1}{r}.(q-1)^{r-1}  \\
&\buildrel{(c)}\over=\frac{q}{(q-1)(n-1)}\sum_{r=1}^{n-1}\binom{n-1}{r}(q-1)^r \\
&=\frac{q\left((1+(q-1))^{n-1}-\binom{n-1}{0}\right)}{(q-1)(n-1)}\\
&=\frac{q^n-q}{(q-1)(n-1)}.
\end{align*}
In $(c)$, we have simplified $\frac{(n-2)!}{(n-r-1)!(r-1)!}.\frac{1}{r} = \frac{1}{n-1}\frac{(n-1)!}{(n-r-1)!r!}.$
By Lemma \ref{lem:duality}, $\frac{q^n-q}{(q-1)(n-1)}$ is an upper bound on $|\Cscr^*_{q,1,n}|.$
 \end{proof}

Although this bound is non-asymptotic, as a corollary we get the asymptotic results of Levenshtein \cite{levenshtein66binary} and Tenengolts \cite{tenengolts84nonbinary}.  
\begin{corollary} The optimal single-deletion correcting code for binary alphabet has size that asymptotically satisfies
\[|\Cscr^*_{2,1,n}| \sim \frac{2^n}{n}.\] 
The optimal single-deletion correcting code for $q$-ary alphabet satisfies 
\[|\Cscr^*_{q,1,n}| \lesssim \frac{q^n}{(q-1)n}.\]
\end{corollary} 
\begin{proof}
For binary alphabet, Levenshtein~\cite{levenshtein66binary} shows that the VT codes correct single deletions. These codes are of size at least $\frac{2^n}{n+1}$, whereby $|\Cscr_{2,1,n}^*| \geq \frac{2^n}{n+1}$. Combining this with Theorem \ref{thm:qary} shows that 
\[\frac{2^n}{n+1} \leq |\Cscr_{2,1,n}^*| \leq \frac{2^n-2}{n-1}.\]
Thus $\frac{|\Cscr^*_{2,1,n}|}{2^n/n} \buildrel{n}\over\rightarrow 1.$ For the $q$-ary case, since by Theorem \ref{thm:qary}, $|\Cscr_{q,1,n}^*| \leq \frac{q^n-q}{(q-1)(n-1)}$, $\lim_{n\rightarrow \infty} \frac{|\Cscr^*_{q,1,n}|}{q^n/n(q-1)} \leq 1.$ 
\end{proof}

\section{Non-asymptotic upper bounds for multiple-deletion correcting codes and the asymptotic rate function}\label{sec:larger}
We now extend the logic used in the bound  above to channels with multiple deletions. 

And as we did in the single-deletion case, we will use the hypergraph $\Hdqsn$ to obtain our bound.
The key property employed in the proof of Theorem \ref{thm:qary} was that the number of runs of a string increases under the insertion of a symbol. This is in fact a specific consequence of a more general property shown by Hirschberg and Regnier~\cite[Lemma 3.1]{hirschberg00tight}: for any $s$, the size of the $s$-deletion set of a string increases under the insertion of a symbol.  This result is articulated in the following lemma. Here if $x=x_1x_2\hdots x_n$ and $y=y_1y_2\hdots y_m$ are $q$-ary strings, `$xy$' denotes the string
$x_1x_2\hdots x_ny_1y_2\hdots y_m.$
\begin{lemma}\label{lem:hirschberg}
Let $s \in \Nbb$. For any strings $x,y \in \Fbb_q^*$ and any symbol $\sigma \in \Fbb_q$, $|D_s(xy)| \leq |D_s(x\sigma y)|$.
\end{lemma}

The original result from \cite[Lemma 3.1]{hirschberg00tight} seems to pertain to nonempty strings $x,y$; this is apparent from their proof. However the extension to the case where one of $x,y$ is empty is trivial and we have included it in the above statement. The consequence is that, in this lemma, $\sigma$ can be thought of as a symbol inserted into an existing string $xy$. A recursive application of 
Lemma \ref{lem:hirschberg} then immediately yields that for any $s$ and any string $x\in \Fbb_q^n$, 
\begin{equation}
|D_s(x)| \leq |D_s(y)|, \qquad \forall\ y \in I_s(x). \label{eq:delmon}
\end{equation}
Looking back at the size of the single-deletion set from \eqref{eq:delsize}, one sees that the monotonicity relationship of Lemma \ref{lem:runs} is a special case of \eqref{eq:delmon}. 

We now exploit \eqref{eq:delmon} to give an upper bound on the size of an $s$-deletion correcting code for arbitrary $s$. The proof utilizes, as before, the fractional transversal problem of $\Hdqsn$.
$$ \problemsmall{$\tau^*(\Hdqsn)=\ $}
        {w}
        {\sum_{x \in \Fbb_q^{n-s}} w(x)}
                                 {\hspace{-1cm}\begin{array}{r@{\ }c@{\ }ll}
\sum_{x \in D_s(y)} w(x) &\geq & 1, & \forall  y\in \Fbb_q^{n},\\
w(x)                                    &\geq & 0,  & \forall x \in \Fbb_q^{n-s}.                              
        \end{array}} 
$$
\begin{theorem}\label{thm:sqary}
Let $s,q,n \in \Nbb$ such that $n>s, q\geq 2$. The optimal $s$-deletion correcting code $\Cscr^*_{q,s,n}$ satisfies
\begin{equation}
 |\Cscr^*_{q,s,n}| \leq \sum_{x \in \Fbb_q^{n-s}}\frac{1}{|D_s(x)|}. \label{eq:sqary}
\end{equation} 
\end{theorem}
\begin{proof}
We construct a fractional transversal for $\Hdqsn$. Consider the candidate fractional transversal $w$, such that for any $x \in \Fbb_q^{n-s}$, $w(x) = \frac{1}{|D_s(x)|}.$ Obviously, $w \geq 0$. Furthermore, for any $y\in \Fbb_q^n,$
\[\sum_{x \in D_s(y)} w(x) = \sum_{x \in D_s(y)}\frac{1}{|D_s(x)|} \buildrel{(a)}\over\geq   1,  \]
where $(a)$ follows from the monotonicity relation \eqref{eq:delmon}. Thus $w$ is indeed a fractional transversal of $\Hdqsn$. Now by Lemma \ref{lem:duality}, the weight of $w$ is an upper bound on $\nu(\Hdqsn) = |\Cscr^*_{q,s,n}|$, whereby 
the result follows.
\end{proof}

In order to derive explicit bounds, we now discuss the sizes of $s$-deletion sets.
For $s \leq 5$, Mercier \etal \cite[Section III.D]{mercier08number} give closed form formulae for the size of $s$-deletion sets, which unlike in the single-deletion case, have quite a complicated form. Closed form expressions for $2$-deletion sets for binary alphabet are also given by Swart and Ferreira~\cite{swart03note} and Sloane~\cite{sloane02single}. 
% Other closed form formulae for binary strings 
The only results on deletion sets valid for arbitrary $s$ are bounds. For all $x \in \Fbb_q^n,$ the $s$-deletion set of $x$ admits the following lower bound, shown recently by Liron and Langberg~\cite[Theorem VI.2]{liron12characterization}. For any $s<n$ and any string $x\in \Fbb_q^n$ with $2<r(x)\leq n,$
\begin{align}
 |D_s(x)| &\geq \delta(r(x),s) + \sum_{i=s+r(x)-n-1}^{\min(s-2,r(x)-3)}\delta(r(x)-2,i)  \label{eq:lllb}\\
\where       &  \delta(r,s) \triangleq \begin{cases}
\sum_{i=0}^s \binom{r-s}{i}, & r> s\geq 0,\\
1, & s=r\geq 0,\\
0, & s<0 \ {\rm or}\ s>r.
\end{cases} \label{eq:delta}
\end{align}
Notice that this bound on $|D_s(\cdot)|$ is always positive. 
Additionally it is an improvement on previous bounds of Levenshtein~\cite{levenshtein92perfect} and Hirschberg and  Regnier~\cite{hirschberg00tight}. 
%: for any $x \in \Fbb_q^n$,
%\begin{equation}
%|D_s(x)| \geq \delta(r(x),s). \label{eq:hirschberg}
%\end{equation}
% All of the discussion from  Section \ref{sec:hyperchar} pertaining to difficulties in obtaining an upper bounds applies equally to the case of larger number of deletions. 

By using the explicit formulae (\eg, \cite{mercier08number,swart03note,sloane02single}) for the sizes of $s$-deletion sets 
in \eqref{eq:sqary}, one may obtain explicit upper bounds on $|\Cscr^*_{q,s,n}|$, for $s\leq 5$. For general $s$, we derive an upper bound on the right hand side of \eqref{eq:sqary} by combining Theorem \ref{thm:sqary} with the lower bound in \eqref{eq:lllb}. Note that the explicit formulae will yield tighter bounds than the one below.
\begin{corollary}\label{cor:sdelub}
Let $s,q,n \in \Nbb, q \geq 2,n > 2s$. The optimal $s$-deletion correcting code $\Cscr^*_{q,s,n}$ satisfies 
$$|\Cscr^*_{q,s,n}| \leq U_{q,s,n},$$ where 
\begin{align}
U_{q,s,n} &\triangleq \sum_{r=3}^{n-s}\frac{q(q-1)^{r-1}\binom{n-s-1}{r-1}}{\delta(r,s)+ \sum_{i=s+r-(n-s)-1}^{\min(s-2,r-3)}\delta(r-2,i)}     \non \\
 &\qquad+ \sum_{r=1}^{2}q(q-1)^{r-1}\binom{n-s-1}{r-1}, \label{eq:explicit}
\end{align}
and $\delta(\cdot,\cdot)$ is as defined in \eqref{eq:delta}. 
\end{corollary}
\begin{proof}
By Theorem \ref{thm:sqary}, we have 
\begin{align*}
|\Cscr^*_{q,s,n}| &\leq \sum_{x \in \Fbb_q^{n-s}: r(x)\geq 3} \frac{1}{|D_s(x)|}+\sum_{x \in \Fbb_q^{n-s}: r(x) < 3} \frac{1}{|D_s(x)|}.   
\end{align*}
For $n-s> s$ and strings $x\in \Fbb_q^{n-s}$  such that $r(x) \geq 3$, the bound in \eqref{eq:lllb} applies; furthermore, notice that for such $x$, the bound in \eqref{eq:lllb} is strictly positive. So using \eqref{eq:lllb} in the equation above, the first sum can be upper-bounded and the resulting bound is the first term in \eqref{eq:explicit}. The second sum in the equation above admits the trivial upper bound $|\{x \in \Fbb_q^{n-s} | r(x) \leq 2\}|,$ which is the second term in \eqref{eq:explicit}. Hence the bound.
\end{proof}

%The above corollary assumes $n>2s.$ Below, we derive a weaker bound valid for any $n>s$.
%\begin{corollary}\label{cor:vqsn}
%Let $s,q,n \in \Nbb, n>s, q \geq 2$. The optimal $s$-deletion correcting code $\Cscr^*_{q,s,n}$ satisfies 
%\[|\Cscr^*_{q,s,n}| \leq V_{q,s,n},\]
%where
%\begin{align}
%V_{q,s,n} &\triangleq \sum_{r=s}^{n-s}\frac{q(q-1)^{r-1}\binom{n-s-1}{r-1}}{\delta(r,s)}    \non \\
%&\qquad + \sum_{r=1}^{s-1}q(q-1)^{r-1}\binom{n-s-1}{r-1} \label{eq:vqsn}
%\end{align}
%\end{corollary}
%\begin{proof}
%The proof follows from the bound in \eqref{eq:hirschberg} and 
%the same arguments as in the proof of Corollary \ref{cor:sdelub}.
%\end{proof}

One of the aims of this paper was to produce non-asymptotic upper bounds that imply known asymptotic bounds. 
We now show that the bound $U_{q,s,n}$  meets this purpose. Our main result is that $U_{q,s,n}$ (and the expression $\sum_{x \in \Fbb_q^{n-s}} \frac{1}{|D_s(x)|}$) implies the previous results of Levenshtein~\cite{levenshtein66binary} stated in \eqref{eq:levenasymp} for $q=2$, and generalizes these results to $q$-ary alphabet.

In order to do this, we first show a \textit{lower} bound on (the upper bound) $U_{q,s,n}.$ For this we recall an upper bound on sizes of deletion sets due to Levenshtein~\cite{levenshtein92perfect}: for any $n,q \in \Nbb$, 
\begin{equation}
|D_s(x)|\leq \binom{r(x)+s-1}{s}, \qquad \forall x\in \Fbb_q^n. \label{eq:levdelub}
\end{equation}

\begin{lemma}\label{lem:uqsnlb}
Let $q,s,n \in \Nbb$, $n>2s, q \geq 2$. The upper bound  $U_{q,s,n}$ satisfies the lower bound
\[ U_{q,s,n} \geq  \sum_{x \in \Fbb_q^{n-s}} \frac{1}{|D_s(x)|} \geq \frac{q^n- q\sum_{r=0}^{s-1} (q-1)^{r} \binom{n-1}{r}}{(q-1)^s\binom{n-1}{s}}.
\]
\end{lemma}
\begin{proof} The first inequality on the left follows from the proof of Corollary \ref{cor:sdelub}. 
To show the second inequality, use the upper bound on $|D_s(\cdot)|$ from \eqref{eq:levdelub}, to get that the sum $\sum_{x \in \Fbb_q^{n-s}} \frac{1}{|D_s(x)|}$ is no less than
\begin{align*}
\sum_{x \in \Fbb_q^{n-s}} &\frac{1}{\binom{r(x)+s-1}{s}} 
= \sum_{r=1}^{n-s} \frac{q(q-1)^{r-1}\binom{n-s-1}{r-1}}{\binom{r+s-1}{s}} \\
%&= \sum_{r=1}^{n-s}\frac{q(q-1)^{r-1}(n-s-1)!(r-1)!s!}{(n-s-r)!(r-1)!(r+s-1)!}\\
%&= \sum_{r=1}^{n-s} \frac{q(q-1)^{r-1}(n-s-1)!s!}{(n-s-r)!(r+s-1)!} \\
%&= \frac{(n-s-1)!s!}{(n-1)!}\sum_{r=1}^{n-s} \frac{q(q-1)^{r-1}(n-1)!}{(n-s-r)!(r+s-1)!}  \\
&\buildrel{(a)}\over= \frac{q}{(q-1)^s\binom{n-1}{s}}\sum_{r=1}^{n-s}(q-1)^{r+s-1} \binom{n-1}{r+s-1}, \\
%&= \frac{q}{(q-1)^s\binom{n-1}{s}}\left[q^{n-1} - \sum_{r=1-s}^0 (q-1)^{r+s-1} \binom{n-1}{r+s-1}\right]  \\
&= \frac{q^n- q\sum_{r=0}^{s-1} (q-1)^{r} \binom{n-1}{r}}{(q-1)^s\binom{n-1}{s}}.
\end{align*}
In $(a)$ we have used that $\frac{\binom{n-s-1}{r-1}}{\binom{r+s-1}{s}} = \frac{\binom{n-1}{r+s-1}}{\binom{n-1}{s}}$.
This proves the claim.
\end{proof}
Notice that the above calculations are a generalization of our proof of the bound on single-deletion correcting codes in Theorem \ref{thm:qary}. 

We now prove the asymptotics of $U_{q,s,n}$ by deriving a matching asymptotic upper bound. 
\begin{theorem} \label{thm:uqsnasymp}
Let $q,s \in \Nbb, q \geq 2$. 
The upper bound on $s$-deletion correcting codes $U_{q,s,n}$ satisfies
\[ U_{q,s,n} \sim \sum_{x \in \Fbb_q^{n-s}}\frac{1}{|D_s(x)|}\sim \frac{s!q^n}{(q-1)^sn^s},     \]
as $n \rightarrow \infty$. Consequently, as $n\rightarrow \infty,$
\[|\Cscr^*_{q,s,n}| \lesssim \frac{s!q^n}{(q-1)^sn^s}.    \]
\end{theorem}
\begin{proof}
Thanks to Lemma \ref{lem:uqsnlb}, to prove the first set of asymptotics, it suffices to show that $U_{q,s,n}\lesssim \frac{s!q^n}{(q-1)^sn^s}$ as $n\rightarrow \infty$.

Fix $r' \in \Nbb,$ $1\leq s \leq r' \leq n-s$. We first claim that $U_{q,s,n}$ satisfies 
\begin{align}
U_{q,s,n} & \leq \sum_{r =r'}^{n-s}\frac{q(q-1)^{r-1}\binom{n-s-1}{r-1}}{\delta(r',s)} \non\\ 
&\qquad + \sum_{r =1}^{r'-1}q(q-1)^{r-1}\binom{n-s-1}{r-1}. \label{eq:u}
\end{align}
To see this, use \eqref{eq:delta}  to conclude 
\[\delta(r,s)+ \sum_{i=s+r-(n-s)-1}^{\min(s-2,r-3)}\delta(r-2,i) \geq \delta(r,s)   \geq \delta(r',s),  \]
for any $r\geq r',$ and thus bound the terms in \eqref{eq:explicit} corresponding to $r\geq r'$. For terms corresponding to $r<r'$, employ 
the trivial bound $\delta(\cdot,\cdot) \geq 1$. 
%Taking $r'=s$,  gives that $U_{q,s,n} \leq V_{q,s,n}.$ For an arbitrary $r'\geq s$, we get the bound on the extreme right. 
Eq \eqref{eq:u} further implies 
\begin{equation}
 U_{q,s,n}\leq \frac{q^{n-s}}{\delta(r',s)} + \sum_{r =1}^{r'-1}q(q-1)^{r-1}\binom{n-s-1}{r-1}. \label{eq:levencomp}    
\end{equation} 
Consider a binomial distribution with parameters $(n-s-1)$ and $\frac{q-1}{q}$. 
The Chernoff bound on the cumulative binomial distribution implies that for $r'-1<\frac{q-1}{q}(n-s-1)$, the sum 
$\sum_{r =1}^{r'-1}q(q-1)^{r-1}\binom{n-s-1}{r-1}$ is no more than 
\[ q^{n-s}\exp\left( -\frac{((n-s-1)\frac{q-1}{q}-r'-2)^2}{2\frac{q-1}{q}(n-s-1)}\right).     \]
Setting $r'={\bf r}= \frac{q-1}{q}(n-s-1) - \sqrt{(n-s-1)\log(n-s-1)}$  in \eqref{eq:levencomp}, using the Chernoff bound and the fact that $\delta({\bf r},s) \sim s!(\frac{q-1}{q})^sn^s$, as $n \rightarrow \infty$, we get
\[U_{q,s,n} \lesssim \frac{s!q^n}{(q-1)^sn^s},\]
as $n \rightarrow \infty$. 
Combining this bound with Lemma \ref{lem:uqsnlb}, we get $U_{q,s,n} \sim \sum_{x \in \Fbb_q^{n-s}}\frac{1}{|D_s(x)|}\sim  \frac{s!q^n}{(q-1)^sn^s}.$ Finally, by Corollary \ref{cor:sdelub}, we get $|\Cscr^*_{q,s,n}| \lesssim \frac{s!q^n}{(q-1)^sn^s}$.
\end{proof}

Note that in addition to clarifying the asymptotics of $U_{q,s,n}$ the above theorem shows that using explicit formulae for $|D_s(\cdot)|$ in \eqref{eq:sqary} does not lead to any improvement over $U_{q,s,n}$ in an asymptotic sense. 

Notice that the right hand side in \eqref{eq:levencomp} closely resembles the expression in Levenshtein's bound from \eqref{eq:leven}. In fact Levenshtein's expression in \eqref{eq:leven} contains the term `$\binom{n-1}{\cdot}$' in place of `$\binom{n-s-1}{\cdot}$', and therefore appears to be weaker than \eqref{eq:levencomp}. However this observation does not directly translate to a proof that our bound $U_{q,s,n}$ is stronger than Levenshtein's bound. This is because the parameter $r$ in \eqref{eq:leven} is allowed to vary between $s-1$ and $n-1$, whereas in \eqref{eq:levencomp}, $r'$ is allowed to vary between $s-1$ and $n-s$. If one could make the deft argument that for any $n,s$, values of $r$ in \eqref{eq:leven} beyond $n-s$ are inconsequential to the comparison of \eqref{eq:leven} with $U_{q,s,n}$, one could establish that $U_{q,s,n}$ is indeed a better bound than Levenshtein's. We have empirically found that this is true; we discuss this in Section \ref{sec:numeric}.

Finally, it is evident that the bound $U_{q,s,n}$, while explicit, is hard to reduce to a closed form for any $s \neq 1$. It appears that the single-deletion case is a unique one which allows for a neat calculation of a closed form expression. 
%Calculation of a closed form expression for $U_{q,s,n}$ is beyond the scope of the current paper.
% Ascertaining the asymptotics of the bound in \eqref{eq:explicit} and the quest for a closed form expression for it are part of ongoing research and beyond 
\subsection{The asymptotic rate function}
Consider the case of a deletion channel where a fraction $\tau \in [0,1]$ of the symbols in a $q$-ary string are deleted. 
Denote by  $R_q(\tau)$ the asymptotic value of the rate of the largest code for this channel, 
\begin{equation}
 R_q(\tau)\triangleq \limn \frac{1}{n}\log_q |\Cscr^*_{q,\tau n,n}|.    \label{eq:rq}
\end{equation} 
We call $R_q(\tau)$ the asymptotic rate function for the deletion channel. 
Very little seems to be known about this function. Levenshtein's non-asymptotic bounds 
from \eqref{eq:leven} only lead to the conclusion $R_2(\tau) \leq 0.7729$ for $\tau \geq 0.0757$~\cite{levenshtein02bounds}. 
In this section we show that our non-asymptotic bound $U_{q,s,n}$ from Corollary \ref{cor:sdelub} allows for a calculation of a finer bound on $R_q(\cdot)$. 

In order to perform this calculation, we need to address some technicalities. Notice that Corollary \ref{cor:sdelub} assumes $n>2s$ to obtain the bound $U_{q,s,n}$. When $s$ was fixed, this restriction was immaterial. But for $s=\tau n$, this restriction means that Corollary \ref{cor:sdelub} can be used only for $\tau <\half$. For $\tau \geq \half,$ we will use the trivial bound 
\begin{equation}
|\Cscr^*_{q,\tau n,n}| \leq \sum_{x \in \Fbb_q^{n-\tau n}} \frac{1}{|D_s(x)|} \leq q^{(1-\tau) n}. \label{eq:trivial}
\end{equation}
Denote by $h_q(x),x\in [0,1]$ the following function
\[h_q(x) = -x\log_q(x)-(1-x)\log_q(1-x) + x\log_q(q-1),\]
and let $h(\cdot)\equiv h_2(\cdot)$, denote the binary entropy function.
\begin{theorem}\label{thm:rate}
Consider the asymptotic rate function $R_q(\cdot)$ defined in \eqref{eq:rq}. 
For $\tau \in [0,\half)$, the asymptotic rate function satisfies
\[R_q(\tau) \leq \max_{\rho \in [0,1-\tau]}N(\rho;\tau)-D(\rho;\tau),\]
where 
\begin{align*}
 N(\rho;\tau) &= (1-\tau) h_q\left( \frac{\rho}{1-\tau} \right), \\
%  D(\rho;\tau) &=\max (D_1(\rho;\tau), D_2(\rho;\tau)) \\
%  D_1(\rho;\tau)&=      \frac{(\rho -\tau) h\left( \min \left( \frac{\tau}{\rho -\tau},\half \right) \right)}{\log_2 q} \\
D(\rho;\tau) &= \max_{m_{\tau,\rho} \leq \mu \leq \min(\tau,\rho)} \frac{(\rho-\mu)h\left( \min \left( \frac{\mu}{\rho-\mu} \right),\half \right)       }{\log_2 q},           
\end{align*}
and $m_{\tau,\rho}=\max(2\tau +\rho -1,0).$ 
For $\tau \in [\half, 1]$, the asymptotic rate function satisfies
\[R_q(\tau) \leq (1-\tau).    \]
\end{theorem}
The proof is standard, but messy. We have relegated it to the Appendix.

Some remarks about this bound on $R_q(\tau)$ are worth noting. Fig \ref{fig:rate} contains plots of this bound pertaining to  various  alphabet sizes for $\tau \in [0,\half)$. For $\tau=0$, $D(\rho;\tau)=0$ and hence $R_q(0) \leq \max_{\rho \in [0,1]}h_q(\rho)=1$, which is expected. Thereafter for small values of $\tau$ (say $\tau \leq 1/10$), one finds that the rate drops quite sharply. For $\tau \geq \half$, the above bound says $R_q(\tau) \leq 1-\tau$ and so 
 $R_q(1) = 0,$ as expected.  One can easily see that this bound on the rate function is superior to Levenshtein's from~\cite{levenshtein02bounds}. 
 
 However there are obvious shortcomings to our bound.  Notice in Fig \ref{fig:rate} that our bound never hits zero for any $\tau \in [0,\half)$; in fact  it becomes zero only for $\tau=1$. Independently of his bound, Levenshtein~\cite{levenshtein02bounds} argues that $R_q(\tau)$ must be zero for all $\tau \geq \frac{q-1}{q}$. Our bound does not imply this property (Levenshtein's bound on the rate function also does not imply this property). Furthermore, in each of the plots in Fig \ref{fig:rate}, our bound shows an \textit{increase} beyond a certain value of $\tau$. The true asymptotic rate function $R_q(\tau)$ must decrease monotonically with $\tau$. This indicates that our bound becomes vacuous after a certain value of $\tau.$

A fascinating lesson in this is that a non-asymptotic bound such as $U_{q,s,n}$ that yields good asymptotics in one regime may not necessarily do so in other regimes.
\begin{figure}\begin{center}
\includegraphics[scale=.26,clip=true, trim=.6in .3in 1.1in 0.6in]{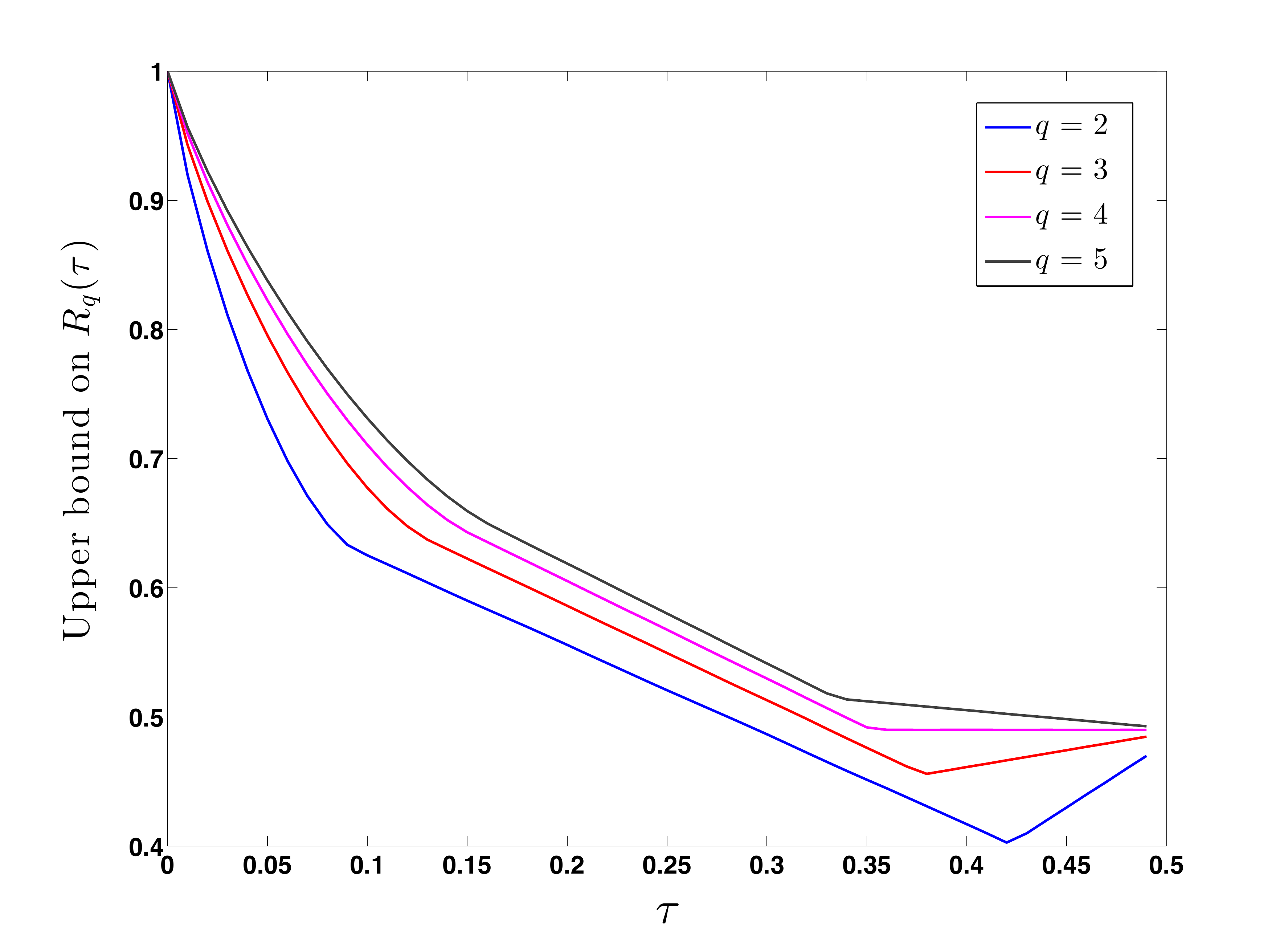}
\end{center}
\caption{The upper on the asymptotic rate function $R_q(\tau)$ guaranteed by Theorem \ref{thm:rate} for alphabet sizes $q=2,\hdots,5$ and $\tau \in [0,\half)$.} \label{fig:rate}
\end{figure}

\section{Bounds on codes for constrained sources} \label{sec:constrained}
The bounds obtained in the previous sections pertain to sizes of codebooks for the set of \textit{all} strings of a particular string length and from a particular alphabet. We now consider the case where a codebook is sought for a constrained set of source strings in $\Fbb_q^n$ and extend the results obtained above to present bounds for such codes.

\begin{definition}
Let $S\subseteq \Fbb_q^n$ be a set of strings and $s\in \Nbb$. An $s$-deletion correcting code or $s$-deletion codebook for $S$, is a subset $C \subseteq S$ such that 
the sets $D_s(x), x \in C$, are pairwise disjoint. The largest such code is denoted $\Cscr^*_{S,s}$ and called the optimal $s$-deletion correcting code or optimal $s$-deletion codebook for $S.$
\end{definition}

Finding a bound on the optimal codebook for an arbitrary set of strings $S$ is significantly more challenging than finding one when $S=\Fbb^n_{q}$. Specifically, arguments 
such as those based on Stirling's approximation employed by Levenshtein~\cite{levenshtein66binary} and Tenengolts~\cite{tenengolts84nonbinary} rely on the availability of all strings in $\Fbb^n_q$. 

We construct our bound by using a suitable hypergraph. Let $S \subseteq \Fbb_q^n$ and define the hypergraph 
\[\Hscr_{S,s}^\Dsf = \left(D_s(S), \{D_s(x): x \in S\}\right),\]
where $D_s(S)=\bigcup_{x\in S} D_s(x)$. $\Hscr_{S,s}^\Dsf$ is the partial hypergraph of $\Hdqsn$ generated by $S$. 
By arguments similar to those previously used, it follows that $\nu(\Hscr_{S,s}^\Dsf)=|\Cscr^*_{S,s}|.$ This matching problem for $\Hscr_{S,s}^\Dsf$ can be explicitly written as follows.
$$ \maxproblemsmall{$|\Cscr^*_{S,s}|=\ $}
        {z}
        {\sum_{y \in S} z(y)}
                                 {\hspace{-1cm}\begin{array}{r@{\ }c@{\ }ll}
\sum_{y \in I_s(x)\cap S} z(y) &\leq & 1, & \forall  x\in D_s(S),\\
z(y)                                    &\in & \Zbb_+,  & \forall y \in S.                              
        \end{array}} 
$$
Notice that in the constraint, the sum is over $y$ belonging to $I_s(x)\cap S$; this is because there may be a case where for some $x \in D_s(S)$, not all strings in $I_s(x)$ are present in $S$, and may thereby not correspond to a hyperedge in $\Hscr_{S,s}^\Dsf$. 
In the language of graphs, the codebook $\Cscr^*_{S,s}$ is a maximum independent set in $L_{S,s}$, the subgraph of $L_{q,s,n}$ induced by strings in $S$. As before, it is easy to see that $L_{S,s}$ is the line graph of $\Hscr_{S,s}^\Dsf$. 

In constructing our bound we exploit the ``decoupling'' afforded by the fractional transversal problem for $\Hscr^\Dsf_{S,s}$. This problem can be explicitly written as follows. 
$$ \problemsmall{$\tau^*(\Hscr^\Dsf_{S,s})=\ $}
        {w}
        {\sum_{x \in D_s(S)} w(x)}
                                 {\hspace{-1cm}\begin{array}{r@{\ }c@{\ }ll}
\sum_{x \in D_s(y)} w(x) &\geq & 1, & \forall  y\in S,\\
w(x)                                    &\geq & 0,  & \forall x \in D_s(S).                              
        \end{array}} 
$$
In this problem there is a separate constraint for each hyperedge, \ie for each string in $S$. Consequently, a fractional transversal can be constructed for $\Hscr^\Dsf_{S,s}$ for any set $S$ by applying the logic used in Theorem \ref{thm:sqary}. 
\begin{theorem} \label{thm:arbit}
Let $q,s,n \in \Nbb, n>s$ and let $S$ be a set of strings in $\Fbb_q^n$. Then 
\begin{equation}
 \frac{|S|}{\binom{n+s-1}{s}\iota_{q,s,n}} \leq |\Cscr^*_{S,s}| \leq \sum_{x \in D_s(S)} \frac{1}{|D_s(x)|}. \label{eq:arbit}
\end{equation}
\end{theorem}
\begin{proof}
Notice that the fractional transversal problem for $\Hscr_{S,s}^\Dsf$ contains a constraint for each string $y$ belonging to $S$ and the sum in this constraint is over all 
$x \in D_s(y)$. Consequently, following Theorem \ref{thm:sqary}, we see that $w(x) = \frac{1}{|D_s(x)|}, x \in D_s(S)$, is a fractional transversal of $\Hscr_{S,s}^\Dsf$. The upper bound thus follows.
 
To obtain the lower bound consider  the line graph $L_{S,s}$ of $\Hscr_{S,s}^\Dsf$. The maximum independent set in $L_{S,s}$ is the optimal matching of $\Hscr_{S,s}^\Dsf$ and thereby the largest codebook $\Cscr_{S,s}^*$. A well known bound given by Brook's theorem or a ``greedy'' algorithm for independent set construction~\cite{west00introduction} gives that 
\[\alpha(L_{S,s}) = |\Cscr^*_{S,s}| \geq \frac{|S|}{\Delta(L_{S,s})+1},\]
where $\Delta(L_{S,s})$ is the maximum degree of a vertex in $L_{S,s}$. The neighborhood of a vertex $x$ in $L_{S,s}$ comprises of those strings obtained 
from $x$ by deletion of $s$ symbols in $x$ followed by the insertion of $s$ symbols in the resulting subsequence. Consequently, $\Delta(L_{S,s}) \leq \max_{x\in S,y\in D_s(S)} |D_s(x)||I_s(y)| -1\leq \binom{n+s-1}{s}\iota_{q,s,n}-1$, where we have used the upper bound on $|D_s(\cdot)|$ from \eqref{eq:levdelub}, $\iota_{q,s,n}$ was defined in \eqref{eq:iqn} as the size of the insertion set for strings in $\Fbb_q^{n-s}$,  and the subtracted $1$ is because the string itself is counted at least once while counting neighbors produced by deletion and insertion. The result follows.
\end{proof}
\subsection{Run-length limited sources}
In this section we will demonstrate the idea above by applying the results of Theorem \ref{thm:arbit} to the specific application of run-length limited codes. For simplicity we consider only the single-deletion case; but the idea is more general and can be extended readily to larger number of deletions. The background on these codes is sourced from the book chapter by Marcus, Roth and Siegel~\cite{marcus98constrained} and their extended monograph available online~\cite{marcus98constrainedmono}. 

Recordings on a magnetic tape when encoded into a binary string result in strings that have no adjacent $1$'s and the number of $0$'s between two consecutive $1$'s is constrained to be in a certain range. Let $0\leq d \leq k$. A binary string is said to satisfy a $(d,k)$-run-length limited (RLL) constraint if a) the string contains no adjacent $1$'s, \ie, the length of any $1$-run is unity,  b) the first and the last runs are $0$-runs and c) the length of any $0$-run is at least $d$ and at most $k$~\cite{marcus98constrained}. In \cite{marcus98constrainedmono}, the first and the last runs of $0$'s are allowed to have lengths less than $d$. In this section we assume, mainly for simplicity, that in a $(d,k)$-RLL string, the first and the last runs of the string must be $0$-runs also having length at least $d$.

 The problem of correcting errors in RLL strings has been considered by several authors (see~\cite[Chapter 9.5]{marcus98constrainedmono}) but most of these works consider erasure error or substitutions (see~\cite{bours94construction} and the discussion therein). Most works that consider deletion, consider the deletion of $0$'s only, since that is most relevant to the application (see, \eg, the discussion in~\cite{schulman99asymptotically}). Recently Cheng \etal~\cite{cheng12moment} and Palun\v{c}i\'{c} \etal~\cite{paluncic12multiple} have considered deletion errors in RLL strings for deletion of 0's and 1's. 

Assume that a set of RLL strings as defined above are to be transmitted through a single-deletion channel, wherein both $0$'s and $1$'s can be deleted. In the theorem below we derive a bound on the size of the largest codebook for a $(d,\infty)$-RLL set of strings. For $0\leq d \leq k$, by $S_n(d,k) \subseteq \Fbb_2^n$ we denote the set of binary strings of length $n$ satisfying the $(d,k)$-RLL constraint. First, we characterize $D_1(S_n(d,\infty))$.
\begin{lemma} \label{lem:sn} Let $n,d\in \Nbb$ and $1 <d \leq n.$ Then we have $D_1(S_n(d,\infty)) = S_{n-1}(d,\infty) \cup S'_{n-1}(d,\infty)$, where 
$S'_{n-1}(d,\infty)$ is the set of binary strings of length $n-1$ such that the first and last runs are $0$-runs, between exactly one pair of consecutive $1$'s there are exactly $d-1$ number of $0$'s and between all other pairs of consecutive $1$'s there are at least $d$ $0$'s.
\end{lemma}
\begin{proof}
``$\subseteq$'': Consider a string in $S_n(d,\infty)$. A deleted symbol must be a $0$ or a $1$. 
%We consider them below. 
\begin{enumerate}
\item If a $0$ is deleted there are two possibilities: either the run from which it is deleted has length $d$, or it has length $>d$. In the former case, the subsequence lies in $S'_{n-1}(d,\infty)$, while in the latter case, it lies in $S_{n-1}(d,\infty)$. 
\item If a $1$ is deleted, the $0$-runs adjacent to the deleted $1$ join to form a longer run of length at least $2d$; the subsequence thus lies in $S_{n-1}(d,\infty)$.
\end{enumerate}
This shows that in either case, $D_1(S_n(d,\infty)) \subseteq S_{n-1}(d,\infty) \cup S'_{n-1}(d,\infty).$

``$\supseteq$'': To show the opposite inclusion, it suffices to show that for any string $x \in S_{n-1}(d,\infty) \cup S'_{n-1}(d,\infty)$ there exists a string $y\in S_n(d,\infty)$ such that $y \in I_1(x)$. Consider an arbitrary $x \in S_{n-1}(d,\infty) \cup S'_{n-1}(d,\infty)$. Insert a $0$ in the shortest $0$-run of $x$ and call the resulting string $y$. Since $x$ has at most one $0$-run of length $d-1$, 
it follows that 
%each $0$-run of $y$ has length at least $d$. Furthermore, since $d>1$, $x$ has no adjacent $1$'s and thereby $y$ has no adjacent $1$'s. Clearly, 
$y$ lies in $S_n(d,\infty)$. 
\end{proof}

Using this lemma and Theorem \ref{thm:arbit}, we will prove an upper bound on the size of a code for $S_n(d,\infty)$. 
\begin{theorem}
Let $n,d \in \Nbb$, $1<d \leq n$. The optimal codebook for $S_n(d,\infty)$, $\Cscr_{S_n(d,\infty)}^*$, satisfies
\begin{align}
|\Cscr_{S_n(d,\infty)}^*| \leq \sum_{r=0}^{\bar{r}}\binom{n-2-r-(d-1)(r+1)}{r}.\frac{1}{2r+1} \non  \\
+ \sum_{r=1}^{\bar{r}'} (r+1)\binom{n-2-r -(d-1)(r+1)}{r-1}. \frac{1}{2r+1}, \label{eq:rll}
\end{align}
where $\bar{r}=\lfloor \frac{n-1-d}{d+1}\rfloor$ and $\bar{r}'=\lfloor\frac{n-d}{d+1}\rfloor$.
\end{theorem}
\begin{proof}
From  \eqref{eq:arbit} and the size of single-deletion sets stated in \eqref{eq:delsize}, $|\Cscr^*_{S_n(d,\infty)}| \leq \sum_{x \in D_1(S_n(d,\infty))} \frac{1}{r(x)}.$ 
%We calculate the right hand side of this inequality. 
By Lemma \ref{lem:sn}, $D_1(S_n(d,\infty)) = S_{n-1}(d,\infty) \cup S'_{n-1}(d,\infty)$. Notice that by definition of $S_{n-1}'(d,\infty),$ the sets $S_{n-1}(d,\infty)$ and $S_{n-1}'(d,\infty)$ are disjoint. Therefore, 
\begin{equation}
|\Cscr^*_{S_n(d,\infty)}|  \leq \sum_{x \in S_{n-1}(d,\infty)} \frac{1}{r(x)}+ \sum_{x\in S_{n-1}'(d,\infty)}\frac{1}{r(x)}. \label{eq:rll2}
\end{equation}
Since all $0$-runs of a string in $S_{n-1}(d,\infty)$ have length at least $d$ and all $1$-runs have unit length, 
and the starting and ending runs are $0$-runs, any string in $S_{n-1}(d,\infty)$ has an odd number of runs and at most $2\bar{r}+1$ runs, where $\bar{r}$ is as stated in the theorem. Therefore a string in $S_{n-1}(d,\infty)$ with, say $2r+1$ runs,  has $r$  $1$-runs of unit length and $r+1$  $0$-runs of lengths say $\ell_1,\hdots,\ell_{r+1}$, where each $\ell_i \geq d.$ The number of strings with $2r+1$ runs in $S_{n-1}(d,\infty)$ is thus equal to the number of integral solutions $(\ell_1,\hdots,\ell_{r+1})$ of  
\begin{align*}
\sum_{i=1}^{r+1} \ell_i = n-1-r, \qquad \ell_i \geq d, 1 \leq i \leq r+1.
\end{align*}
By Lemma \ref{lem:eqsol} this number is $\binom{n-2-r-(d-1)(r+1)}{r}$, whereby the first term in the right hand side of \eqref{eq:rll2} equals the first term in the right hand side of \eqref{eq:rll}. 

Each string in $S'_{n-1}(d,\infty)$ also has odd number of runs. Furthermore, it has at least three runs and at most $2\bar{r}'+1$ runs, where $\bar{r}'$ is defined in the statement of the theorem. Consider a string with $2r+1$ runs with $r$ $1$-runs and $r+1$ $0$-runs. First choose the $0$-run with length $d-1$; this can be chosen in $r+1$ ways. Let $\ell_1,\hdots,\ell_{r}$ be the lengths of the remaining $0$-runs. The number of choices for the lengths of the remaining runs is the number of integral solutions of
\[\sum_{i=1}^{r} \ell_i = n-1-r-(d-1), \qquad \ell_i \geq d, 1 \leq i \leq r.\]
Using Lemma \ref{lem:eqsol}, the number of strings in $S_{n-1}'(d,\infty)$ with $2r+1$ runs is thus $(r+1)\binom{n-2-r -(d-1)r -(d-1)}{r-1}$. This proves that the 
second term in \eqref{eq:rll} equals its counterpart in \eqref{eq:rll2}.
\end{proof}
Unfortunately, calculating these bounds in a simplified closed form does not appear to be easy. Our aim in this section was only to demonstrate the idea and the bound in Theorem \ref{thm:arbit}. Exact calculation of these bounds is beyond the scope of this paper. 

With this we conclude the theoretical portion of the paper. In the following sections we will study how our bounds compare numerically with the sizes of known codebooks and with other bounds.
\section{Numerical results}\label{sec:numeric}
\begin{table}[htbp]
\centering
\subfloat[$q=2$, binary\label{tab:2}]{
\begin{tabular}{c|cccc}
$n$ & $\lfloor \texttt{Lev-UB}\rfloor$ & $\lfloor \frac{2^n-2}{n-1}\rfloor$ & $\lfloor \texttt{LP-UB}\rfloor$ & $|{\rm VT}_0(n)|$ \\[1ex]
\hline
1	&	1	&	--	&	1	&	1	\\
2	&	3	&	2	&	2	&	2	\\
3	&	4	&	3	&	2	&	2	\\
4	&	6	&	4	&	4	&	4	\\
5	&	10	&	7	&	6	&	6	\\
6	&	18	&	12	&	10	&	10	\\
7	&	34	&	21	&	17	&	16	\\
8	&	58	&	36	&	30	&	30	\\
9	&	103	&	63	&	53	&	52	\\
10	&	190	&	113	&	96	&	94	\\
11	&	363	&	204	&	175	&	172	\\
12	&	646	&	372	&	321	&	316	\\
13	&	1182	&	682	&	593	&	586	\\
14	&	2232	&	1260	&	1104	&	1096	
\end{tabular}
}

\subfloat[$q=3$]{
\centering
\begin{tabular}{c|cccc}
$n$ & $\lfloor\texttt{Lev-UB} \rfloor$ &$\lfloor \frac{q^n-q}{(n-1)(q-1)}\rfloor$ & $\lfloor \texttt{LP-UB}\rfloor$ & $|\texttt{Tenengolts}|$ \\[1ex]
\hline
1	&	1	&	--	&	1	&	1	\\
2	&	4	&	3	&	3	&	2	\\
3	&	7	&	6	&	5	&	5	\\
4	&	16	&	13	&	12	&	8	\\
5	&	43	&	30	&	24	&	17	\\
6	&	114	&	72	&	62	&	46	\\
7	&	282	&	182	&	153	&	105	\\
8	&	774	&	468	&	402	&	278	
\end{tabular}
\label{tab:3}
}

\centering
\subfloat[$q=4$ \label{tab:4}]{\begin{tabular}{c|cccc}
$n$ & $\lfloor \texttt{Lev-UB}\rfloor$ &$\lfloor \frac{q^n-q}{(n-1)(q-1)}\rfloor$ & $\lfloor \texttt{LP-UB} \rfloor$ & $|\texttt{Tenengolts}|$ \\[1ex]
\hline
1	&	1	&	--	&	1	&	1	\\
2	&	6	&	4	&	4	&	3	\\
3	&	12	&	10	&	8	&	6	\\
4	&	36	&	28	&	25	&	20	\\
5	&	132	&	85	&	69	&	52	\\
6	&	405	&	272	&	231	&	178	
\end{tabular}}

\subfloat[$q=5$ \label{tab:5}]{\begin{tabular}{c|cccc}
$n$ & $\lfloor \texttt{Lev-UB}\rfloor$ &$\lfloor \frac{q^n-q}{(n-1)(q-1)}\rfloor$ & $\lfloor \texttt{LP-UB}\rfloor$ & $|\texttt{Tenengolts}|$ \\[1ex]
\hline
1	&	1	&	--	&	1	&	1	\\
2	&	7	&	5	&	5	&	3	\\
3	&	17	&	15	&	11	&	9	\\
4	&	67	&	51	&	45	&	33	\\
5	&	293	&	195	&	158	&	129	\\
6	&	1146	&	781	&	657	&	527	
\end{tabular}
}
\caption{The columns of the table show, from left to right, the value of Levenshtein's bound from \eqref{eq:leven} (\texttt{Lev-UB}), values of upper bound obtained in Theorem \ref{thm:qary}, the fractional matching number $\nu^*(\Hdqn)$ (\texttt{LP-UB}),  and the sizes of best known codes, for values of $q$ and $n$. For binary alphabet, the best known codes are the Varshamov-Tenengolts codes ${\rm VT}_0(n)$~\cite{varshamov65codes,levenshtein66binary}. For larger alphabet, the best codes known to us are those of Tenengolts~\cite{tenengolts84nonbinary}, whose size is denoted $|\texttt{Tenengolts}|$.} \label{table}
\end{table}

Recall that the upper bounds guaranteed by Theorems  \ref{thm:qary}, \ref{thm:sqary} and \ref{thm:arbit} were obtained by constructing a fractional transversal for the hypergraphs involved. To obtain an upper bound on the size of optimal codebooks for the deletion channel, it suffices to find the fractional matching number itself, and ideally one would like to have an expression for this number. We were not able to find such an expression and constructed a fractional transversal as a proxy for it. 

In the case of a single deletion, there already exist codes which are known to be asymptotically good. This motivates a comparison between our bound for single-deletion correcting codes, the fractional matching number and the sizes of the best known codes in order to ascertain the quality of these codes. 
To do this, the fractional matching problem for hypergraph  $\Hdqn$ (for single deletions) was solved numerically on \Matlab for various values of $q$ and $n$. 
Table \ref{table} documents the results obtained. 

In each subtable of Table \ref{table}, the columns contain from left to right, the string length $n$, Levenshtein's upper bound (strongest one from \eqref{eq:leven}; denoted \texttt{Lev-UB}), the bound from Theorem \ref{thm:qary}, the value of the fractional matching number found numerically ($=\nu^*(\Hdqn);$ denoted \texttt{LP-UB}), and the best known code for each case. In the binary case the best known code is the Varshamov-Tenengolts code ${\rm VT}_0(n)$ where 
\[{\rm VT}_a(n) = \left\{x_1x_2\hdots x_n\in \Fbb_2^n \left\lvert \sum_i ix_i =a \bmod n+1\right.\right\}.\]
${\rm VT}_0(n)$  is also conjectured~\cite{sloane02single} to be optimal for all $n$. 
 For larger alphabet the best codes we know of are those of Tenengolts~\cite{tenengolts84nonbinary} (these are denoted $|\texttt{Tenengolts}|$). 
For each $q$ the largest value of $n$ is as far as we could compute with the resources available to us.

The first trend noticeable is that in any row values decrease from left to right. Thus the strongest of Levenshtein's bounds from \eqref{eq:leven} is weaker than our non-asymptotic bound. Our non-asymptotic bound is also weaker than the value of the fractional matching number (column \texttt{LP-UB}); this shows that the fractional transversal we have constructed to obtain the upper bound is not the optimal fractional transversal. 

Notice that in the binary case, shown in Table \ref{tab:2}, the size of the Varshamov-Tenengolts code ${\rm VT}_0(n)$ shows a good match with 
with \texttt{LP-UB}. This indicates that these codes are either optimal (as conjectured) or close to being optimal, at least for $n\leq 14$. Sloane's website~\cite{sloaneweb} carries numerically obtained bounds for $n\leq 11,$ 
of which ${\rm VT}_0(n)$ has been confirmed as optimal for $n\leq 10$.  The bounds on the website have been obtained by computing the Lov\'{a}sz $\vartheta$~\cite{west00introduction} on graphs $L_{q,1,n}$. The results in Table \ref{table} may be considered as additions to Sloane's compilation. 

For each value of $q,n,$ Tenengolts' construction gives a two-parameter family of codes (the parameters being $\beta,\gamma$ in \cite[Eq (2)]{tenengolts84nonbinary}). The column $\texttt{|Tenengolts|}$ contains for the respective $q,n,$ the largest code out of this family. Unlike in the VT codes where it is known that of the family ${\rm VT}_a(n), a =0,\hdots,n$, the code ${\rm VT}_0(n)$ is the largest, we are not aware of a similar characterization of the largest code from Tenengolts' family. Thus the column $\texttt{|Tenengolts|}$ was populated by explicitly calculating the size of the code for each value of the parameters and thereafter identifying the largest of those. It is clear from this table that these codes are quite smaller than the fractional matching number in \texttt{LP-UB}. This may mean either that there is a large gap between the fractional matching number and the matching number for these hypergraphs, or that the Tenengolts codes are not optimal.

For larger number of deletions there exist no good codes apart from those found by search. So no interesting comparisons can be made for an existing code for a larger number of deletions. However, we may compare our bound with Levenshtein's from \eqref{eq:leven}. Figure \ref{fig:levencomp} shows the comparison for binary alphabet and $s=2,3,4$ and $15\leq n \leq 30$. We have focused on this region of $n$ so as to allow the distinctions between the lines for $s=2,3,4$ coming from Levenshtein's bound to be clearly discerned; for smaller values of $n$ these lines overlap. One can easily eye-ball that our bound is significantly better than Levenshtein's. 
\begin{figure}
\begin{center}
\includegraphics[scale=.27, clip=true, trim=.8in .4in 1.15in .3in]{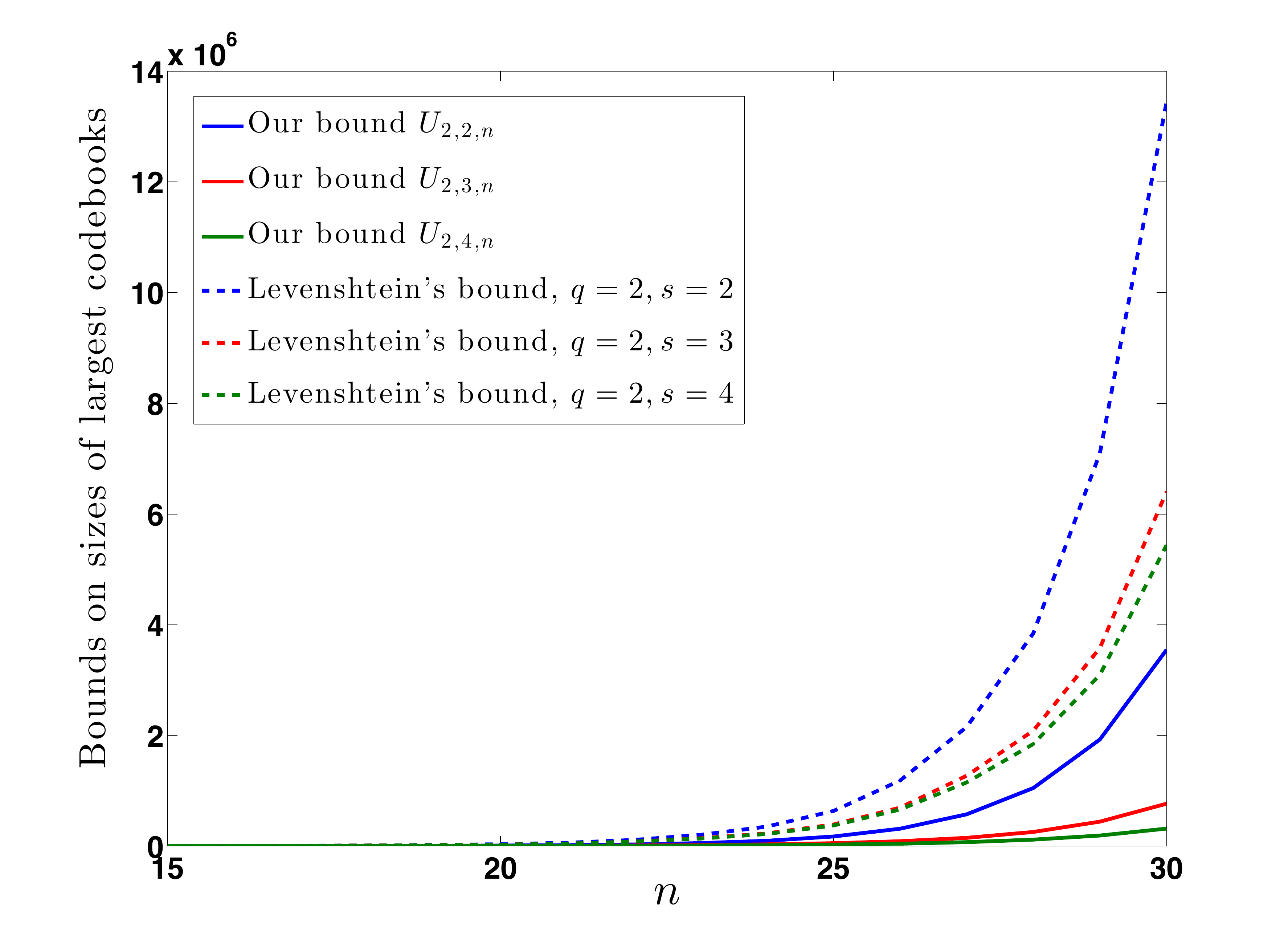}
\end{center}
\caption{Figure showing values of $U_{q,s,n}$ (solid lines) and Levenshtein's bound (dotted lines) from \eqref{eq:leven} for $q=2$, $s=2,3,4$ and $15\leq n \leq 30$.}\label{fig:levencomp}
\end{figure}

We discuss the quality of  our bound and prospects for improving it in the next section.

\section{Discussion}\label{sec:disc}
For the sake of this discussion, we limit ourselves to the case of the single-deletion channel. 
Table \ref{table} shows that there is scope for improving our bound $\frac{q^n-q}{(q-1)(n-1)}$ for the $q$-ary single-deletion channel. Since the bound is not equal to the fractional matching number \texttt{LP-UB},  one can obtain a better bound by merely finding a fractional transversal with a smaller weight. However, in practice a construction to this effect has eluded us. In fact, our constructed transversal shows a close match to the optimal fractional transversal found numerically, which makes any improvement challenging. We discuss this below.

Figure \ref{fig:f27} shows the optimal fractional transversal and the fractional transversal we have constructed ($w(\cdot) \equiv \frac{1}{r(\cdot)}$) for hypergraph $\Hscr_{2,1,n}^\Dsf$, \ie $q=2, n=8$ and $s=1$ and for hypergraph $\Hscr_{5,1,4}^\Dsf$ ($q=5,n=4, s=1$). Notice that in both cases, the constructed fractional transversal matches the general trend of the optimal fractional transversal. 
This continues to hold \remove{to} for larger values of $n$. Indeed, in the binary case, since 
\[0 \leq \frac{\frac{2^n-2}{n-1} - \nu^*(\Hscr_{2,1,n}^\Dsf)}{2^{n-1}} \leq \frac{\frac{2^n-2}{n-1} - \frac{2^n}{n+1}}{2^{n-1}} \rightarrow 0,\]
the average difference between the constructed and optimal transversal vanishes for large $n$. A tighter bound may be obtained by fine-tuning the constructed fractional transversal, but since the general trend of the optimal fractional transversal has already been captured by our constructed transversal, the logic for further fine-tuning is not obvious. Yet, this effort is not a lost cause: since the number of vertices grows 
exponentially, a small saving in this construction may imply a substantial improvement in the bound. 

\begin{figure}[htbp]
\centering
\subfloat[$\Hscr_{2,1,8}^\Dsf$]{\includegraphics[scale=.27, clip=true, trim = .9in .3in 1.1in 0.7in]{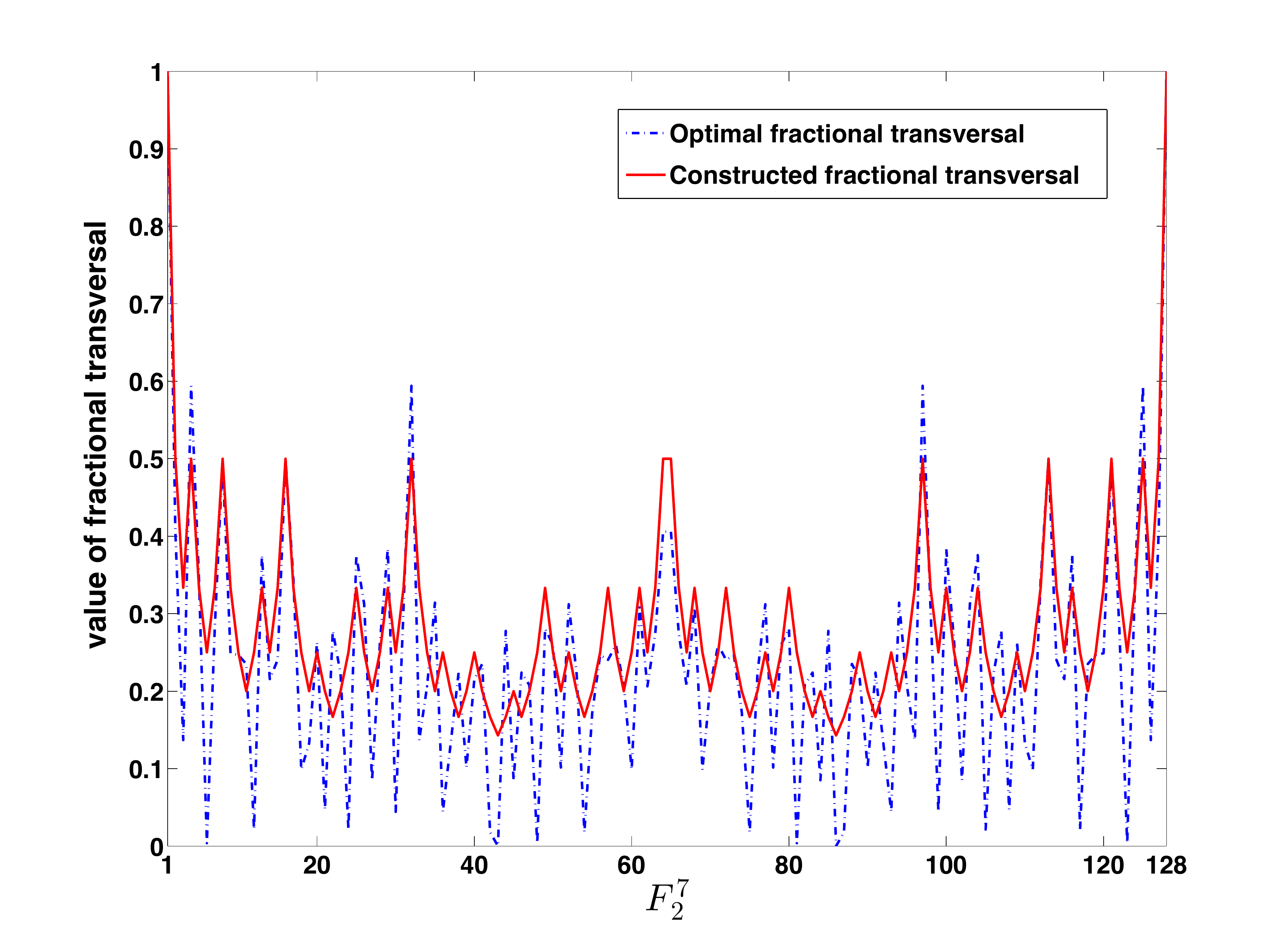}}

\subfloat[$\Hscr_{5,1,4}^\Dsf$]{\includegraphics[scale=.27, clip=true, trim = 0.9in .18in 1.1in .7in]{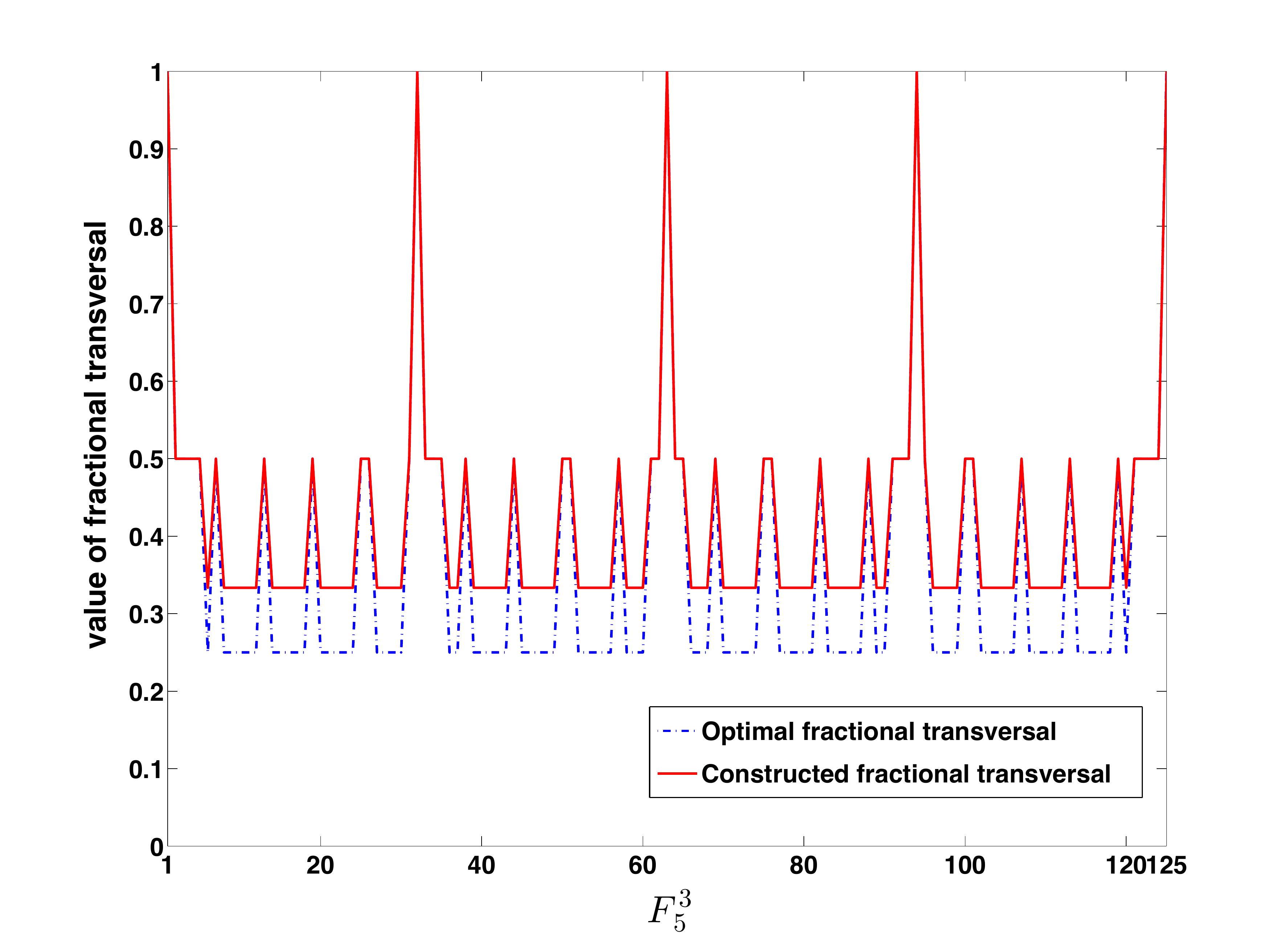}}
\caption{The horizontal axis consists of elements of $\mathbb{F}_2^7$ and $\Fbb_5^3,$ respectively, plotted in increasing order of their decimal value. The vertical axis is the value of the fractional transversals. In each case, the dotted line shows the optimal fractional transversal and the solid line shows the constructed fractional transversal $w(x) \equiv \frac{1}{r(x)}$ for $\Hscr_{2,1,8}^\Dsf$ and $\Hscr_{5,1,4}^\Dsf,$ respectively. These lines are provided to aid in discerning the trends in their values; they have no meaning per se.}
\label{fig:f27}
\end{figure}
\note{$\Fbb_5^4 \rightarrow \Fbb_5^3$?}
We end with one final consideration and speculate on what may be an alternative approach to obtaining better bounds. Since the most successful approaches to code construction for this problem have been number-theoretic one may be inclined to conjecture that the size of the optimal codebook $|\Cscr_{q,1,n}^*|$ depends not only on the numerical value of $n$, but also on properties $n$ has as a number. In the binary case, in particular, since the fractional matching number $\nu^*(\Hscr_{2,1,n}^\Dsf)$ closely tracks $|{\rm VT}_0(n)|$, which is given by a number-theoretic formula (see \cite[Eq (7)]{sloane02single}), it appears that $\nu^*(\Hscr_{2,1,n}^\Dsf)$ may also be given by a number-theoretic expression. In contrast, neither our bounds nor  their proofs have any number-theoretic character. Perhaps a clue to tightening these bounds lies in giving a number-theoretic construction of the optimal fractional matching or a better (possibly optimal) fractional transversal.

%The fractional transversal emerged as a natural concept due to the hypergraph formulation we have chosen. It is therefore not clear if the bound $q\sum_{r=1}^{n-1} \binom{n-2}{r-1}.\frac{(q-1)^{r-1}}{r}$ we have obtained is an artifact of this particular formulation or if it has a more fundamental interpretation for the problem. A related question is that of tightening these bounds.  
%\section{Conclusion} \label{sec:conc}
In summary, this paper considered the deletion channel for general $q$-ary alphabet and an arbitrary number of deletions and proved new non-asymptotic upper bounds on the sizes of the optimal codebooks. The bounds are stronger than known bounds and imply classical asymptotic bounds. The bounds were derived via a hypergraph characterization of the optimal codebook and a linear programming argument. The approach was extended to derive bounds on codebooks for general constrained sources and was demonstrated for run-length limited sources. The paper concluded with a discussion on numerical results and on the quality of these bounds.

\appendix[Proof of Theorem \ref{thm:rate}] \label{app:rate}
\begin{proof} First consider $\tau \in [0,\half).$ \\ 
For such a value of $\tau$, the bound \eqref{eq:explicit} applies. 
By \eqref{eq:explicit},
\begin{align*}
U_{q,\tau n,n} &= \sum_{r=3}^{(1-\tau) n}\frac{q(q-1)^{r-1}\binom{(1-\tau) n-1}{r-1}}{\delta(r,\tau n)+ \sum_{i=(2\tau-1) n+r-1}^{\min(\tau n-2,r-3)}\delta(r-2,i)}     \non \\
 &\qquad+ \sum_{r=1}^{2}q(q-1)^{r-1}\binom{(1-\tau) n-1}{r-1}. 
\end{align*}
Notice that the second sum being a mere polynomial in $n$ can be ignored in comparison to the first sum. Below, we focus only on the first term and estimate its asymptotics by finding its exponent.

Put $r= \rho n$ so that $\rho \in [0,1-\tau]$, and let \begin{align}
N(\rho;\tau)&=\limn \frac{1}{n}\log_q q(q-1)^{\rho n-1}\binom{(1-\tau) n-1}{\rho n-1},\non\\
D_1(\rho;\tau) &= \limn \frac{1}{n}\log_q\delta(\rho n,\tau n), \non \\
D_2(\rho;\tau) &= \limn \frac{1}{n}\log_q \sum_{i=(2\tau -1 + \rho)n-1}^{\min(\tau n-2,\rho n-3)} \delta(\rho n-2,i).\non 
% &= \frac{1}{n}+\frac{\rho[(1-\tau)n-1]}{n} \log_q (q-1) + \frac{(1-\tau) n-1}{n} \frac{1}{(1-\tau) n-1}\log_q \binom{(1-\tau) n-1}{\rho [(1-\tau)n-1]}\\
% \\
% &\delta(\rho[(1-\tau)n-1],\tau n) = \sum_{i=0}^{\tau n} \binom{\rho[(1-\tau)n-1] - \tau n}{i}
\end{align}
Here $N(\rho;\tau)$ is the exponent of the numerator and the exponent of the denominator is 
$$D(\rho;\tau) =\max (D_1(\rho;\tau), D_2(\rho;\tau)).$$
Therefore, the asymptotic rate function satisfies
\[R_q(\tau) \leq \max_{0 \leq \rho \leq 1-\tau} N(\rho;\tau) - D(\rho;\tau).\]
We now calculate the above exponents. It is easy to see that $$N(\rho;\tau)= (1-\tau)h_q\left(\frac{\rho}{1-\tau}\right), $$ 
which is as required.
% For the denominator,  where $D_1(\rho;\tau)$ is the exponent of $\delta(\rho n,\tau n)$ and $D_2(\rho;\tau)$ is the exponent of $\sum_{i=(2\tau-1) n+\rho n-1}^{\min(\tau n-2,\rho n -3)}\delta(\rho n-2,i)$. 
Next consider $D_1(\rho;\tau)$. Clearly, if $\rho\leq \tau,$ $D_1(\rho;\tau)=0$. If $\tau \leq \frac{\rho-\tau}{2},$ \ie, $\rho \geq 3\tau$, 
\begin{align*}
% = \frac{\rho[(1-\tau)n-1] - \tau n}{n}\frac{1}{\rho[(1-\tau)n-1] - \tau n}\frac{\log_2 \delta(\rho[(1-\tau)n-1],\tau n)}{\log_2q} \\
D_1(\rho;\tau)= (\rho-\tau) \frac{h\left(\frac{\tau}{\rho -\tau}\right)}{\log_2 q}.
\end{align*}
On the other hand if $\rho < 3\tau,$ $D_1(\rho;\tau)= \frac{\rho-\tau}{\log_2 q}.$ In summary, we get
\begin{align*}
  D_1(\rho;\tau)&=      \I{\rho>\tau}\left(\frac{(\rho -\tau) h\left( \min \left( \frac{\tau}{\rho -\tau},\half \right) \right)}{\log_2 q}\right).  \end{align*}
Now consider $D_2(\rho;\tau)$. 
Recall from \eqref{eq:delta} that if $i <0, \delta(\rho n -2,i)=0$.  In the expression for $D_2(\rho;\tau)$, put $i = \mu n,$ so that $\mu \in [\max(2\tau +\rho -1,0), \min(\tau,\rho)].$ Then arguing as above, we get 
% If $\rho \leq 1-2\tau,$ 
% \begin{align*}
%  D_2(\rho;\tau) &= \max_{0 \leq \mu \leq \min(\tau,\rho)} \frac{(\rho-\mu)h\left( \min \left( \frac{\mu}{\rho-\mu} \right),\half \right)       }{\log_2 q}.           
% \end{align*}
% If $\rho > 1-2\tau,$
\begin{align*}
 D_2(\rho;\tau) &= \max_{m_{\tau,\rho}\leq \mu \leq \min(\tau,\rho)} \frac{(\rho-\mu)h\left( \min \left( \frac{\mu}{\rho-\mu} \right),\half \right)       }{\log_2 q},           
\end{align*}
where $m_{\tau,\rho}=\max(2\tau +\rho -1,0),$ as stated in the theorem.

We now show that $D_2(\rho;\tau)$ dominates $D_1(\rho;\tau)$ for any $\rho,\tau$. If $\rho \leq \tau, D_1(\rho;\tau) \equiv 0$, 
so, clearly, $D_2(\rho;\tau) \geq D_1(\rho;\tau)$. However, if $\rho >\tau$, we find that  $\mu =\tau$ satisfies $\mu \in [m_{\tau,\rho}, \min(\tau,\rho)].$ To see this, observe that a) $\min(\tau,\rho) =\tau$, since $\rho>\tau$, and b) $\tau \geq m_{\tau,\rho}$ if and only if $\rho \leq 1-\tau$, which is the assumed range on $\rho.$ But for $\mu =\tau$ the value of the maximand above equals $D_1(\rho;\tau).$ Consequently, $D_2(\rho;\tau)$, which involves a maximization over $\mu$,  dominates $D_1(\rho;\tau).$ In summary, \[D(\rho;\tau) = D_2(\rho;\tau),    \]
as required. This completes the first part of the theorem pertaining to $\tau \in [0,\half)$.

Now consider $\tau \geq \half$ and use the trivial bound from \eqref{eq:trivial}. In this case, clearly,
$$R_q(\tau) \leq (1-\tau).$$ 
% 
% 
% and $i \sim \mu n$, $\mu \in [0,1]$. If $\mu \leq \frac{\rho(1-\tau)}{3m_{\tau,\rho}} = \frac{1}{3\min(\frac{\tau}{\rho(1-\tau)},1)},$
% \begin{align*}
%  \frac{1}{n}\log_q &\delta(\rho((1-\tau)n-1)-3,\mu m_{\tau,\rho}n) \\ 
%  &\asn \left(\rho(1-\tau) -\mu \right) \frac{h\left(\frac{\mu}{\rho(1-\tau) -\mu}\right)}{\log_2 q}
% \end{align*}
% If $\mu > \frac{1}{3\min(\frac{\tau}{\rho(1-\tau)},1)},$
% \[\frac{1}{n}\log_q \delta(\rho((1-\tau)n-1)-3,\mu m_{\tau,\rho}n) \asn \frac{\rho(1-\tau) -\mu }{\log_2 q}\]
% 
% 
This covers all cases and the proof is complete.
\end{proof}

\bibliographystyle{IEEEtran.bst} 
\bibliography{ref}

% Generated by IEEEtran.bst, version: 1.13 (2008/09/30)
\begin{thebibliography}{10}
\providecommand{\url}[1]{#1}
\csname url@samestyle\endcsname
\providecommand{\newblock}{\relax}
\providecommand{\bibinfo}[2]{#2}
\providecommand{\BIBentrySTDinterwordspacing}{\spaceskip=0pt\relax}
\providecommand{\BIBentryALTinterwordstretchfactor}{4}
\providecommand{\BIBentryALTinterwordspacing}{\spaceskip=\fontdimen2\font plus
\BIBentryALTinterwordstretchfactor\fontdimen3\font minus
  \fontdimen4\font\relax}
\providecommand{\BIBforeignlanguage}[2]{{%
\expandafter\ifx\csname l@#1\endcsname\relax
\typeout{** WARNING: IEEEtran.bst: No hyphenation pattern has been}%
\typeout{** loaded for the language `#1'. Using the pattern for}%
\typeout{** the default language instead.}%
\else
\language=\csname l@#1\endcsname
\fi
#2}}
\providecommand{\BIBdecl}{\relax}
\BIBdecl

\bibitem{sloane02single}
N.~J.~A. Sloane, ``On single-deletion-correcting codes,'' in \emph{Codes and
  Designs: Proceedings of a Conference Honoring {P}rofessor {Dijen K.}
  {Ray-Chaudhuri} on the Occasion of His 65$\th$ Birthday, {The Ohio State
  University, May 18-21, 2000}}.\hskip 1em plus 0.5em minus 0.4em\relax Walter
  de Gruyter, 2002.

\bibitem{levenshtein66binary}
V.~I. Levenshtein, ``Binary codes capable of correcting deletions, insertions,
  and reversals,'' \emph{Soviet Physics Doklady}, vol.~10, no.~8, pp. 707--710,
  1966.

\bibitem{varshamov65codes}
R.~R. Varshamov and G.~M. Tenengolts, ``Codes which correct single asymmetric
  errors (in {R}ussian),'' \emph{Avtomatika i Telemekhanika}, vol.~6, no.~2,
  1965.

\bibitem{sloaneweb}
\BIBentryALTinterwordspacing
N.~J.~A. Sloane, ``Challenge problems: Independent sets in graphs,'' Jul. last
  updated 2011. [Online]. Available:
  \url{http://neilsloane.com/doc/graphs.html}
\BIBentrySTDinterwordspacing

\bibitem{tenengolts84nonbinary}
G.~M. Tenengolts, ``Nonbinary codes, correcting single deletion or insertion,''
  \emph{Information Theory, {IEEE} Transactions on}, vol.~30, no.~5, pp. 766 --
  769, Sep. 1984.

\bibitem{levenshtein02bounds}
V.~I. Levenshtein, ``Bounds for deletion/insertion correcting codes,'' in
  \emph{2002 {IEEE} International Symposium on Information Theory, 2002.
  Proceedings}, Lausanne, Switzerland, 2002, p. 370.

\bibitem{huffman98handbook}
V.~S. Pless and W.~C. Huffman, Eds., \emph{Handbook of Coding Theory, Volume
  {II}}, 1st~ed.\hskip 1em plus 0.5em minus 0.4em\relax North Holland, Nov.
  1998.

\bibitem{mercier10survey}
H.~Mercier, V.~Bhargava, and V.~Tarokh, ``A survey of error-correcting codes
  for channels with symbol synchronization errors,'' \emph{{IEEE}
  Communications Surveys Tutorials}, vol.~12, no.~1, pp. 87 --96, 2010.

\bibitem{ullman67capabilities}
J.~Ullman, ``On the capabilities of codes to correct synchronization errors,''
  \emph{{IEEE} Transactions on Information Theory}, vol.~13, no.~1, pp. 95
  --105, Jan. 1967.

\bibitem{ullman_near-optimal_1966}
------, ``Near-optimal, single-synchronization-error-correcting code,''
  \emph{{IEEE} Transactions on Information Theory}, vol.~12, no.~4, pp. 418 --
  424, Oct. 1966.

\bibitem{helberg02multiple}
A.~Helberg and H.~Ferreira, ``On multiple insertion/deletion correcting
  codes,'' \emph{{IEEE} Transactions on Information Theory}, vol.~48, no.~1,
  pp. 305 --308, Jan. 2002.

\bibitem{abdel-ghaffar12helbergs}
K.~Abdel-Ghaffar, F.~Palun\v{c}i\'{c}, H.~Ferreira, and W.~Clarke, ``On
  {H}elberg's generalization of the levenshtein code for multiple
  {Deletion/Insertion} error correction,'' \emph{{IEEE} Transactions on
  Information Theory}, vol.~58, no.~3, pp. 1804 --1808, Mar. 2012.

\bibitem{calabi69general}
L.~Calabi and W.~Hartnett, ``Some general results of coding theory with
  applications to the study of codes for the correction of synchronization
  errors,'' \emph{Information and Control}, vol.~15, no.~3, pp. 235--249, Sep.
  1969.

\bibitem{tanaka76synchronization}
E.~Tanaka and T.~Kasai, ``Synchronization and substitution error-correcting
  codes for the {L}evenshtein metric,'' \emph{{IEEE} Transactions on
  Information Theory}, vol.~22, no.~2, pp. 156 -- 162, Mar. 1976.

\bibitem{varshamov73class}
R.~R. Varshamov, ``A class of codes for asymmetric channels and a problem from
  the additive theory of numbers,'' \emph{{IEEE} Transactions on Information
  Theory}, vol.~19, no.~1, pp. 92 -- 95, Jan. 1973.

\bibitem{butenko02finding}
S.~Butenko, P.~Pardalos, I.~Sergienko, V.~Shylo, and P.~Stetsyuk, ``Finding
  maximum independent sets in graphs arising from coding theory,'' in
  \emph{Proceedings of the 2002 {ACM} symposium on Applied computing}, ser.
  {SAC} '02.\hskip 1em plus 0.5em minus 0.4em\relax New York, {NY}, {USA}:
  {ACM}, 2002, p. 542–546.

\bibitem{schulman99asymptotically}
L.~Schulman and D.~Zuckerman, ``Asymptotically good codes correcting
  insertions, deletions, and transpositions,'' \emph{{IEEE} Transactions on
  Information Theory}, vol.~45, no.~7, pp. 2552 --2557, Nov. 1999.

\bibitem{khajouei11algorithmic}
F.~Khajouei, M.~Zolghadr, and N.~Kiyavash, ``An algorithmic approach for
  finding deletion correcting codes,'' in \emph{2011 {IEEE} Information Theory
  Workshop {(ITW)}}, Paraty, Brazil, Oct. 2011, pp. 25 --29.

\bibitem{cullina12coloring}
D.~Cullina, A.~A. Kulkarni, and N.~Kiyavash, ``A coloring approach to
  constructing deletion correcting codes from constant weight subgraphs,'' in
  \emph{Proceedings of the ISIT}, Cambridge, USA, 2012.

\bibitem{roth94lee}
R.~Roth and P.~Siegel, ``Lee-metric {BCH} codes and their application to
  constrained and partial-response channels,'' \emph{{IEEE} Transactions on
  Information Theory}, vol.~40, no.~4, pp. 1083 --1096, Jul. 1994.

\bibitem{hilden91shift}
H.~Hilden, D.~Howe, and J.~Weldon, {E.J.}, ``Shift error correcting modulation
  codes,'' \emph{{IEEE} Transactions on Magnetics}, vol.~27, no.~6, pp. 4600
  --4605, Nov. 1991.

\bibitem{bours94construction}
A.~Bours, ``Construction of fixed-length insertion/deletion correcting
  runlength-limited codes,'' \emph{{IEEE} Transactions on Information Theory},
  vol.~40, no.~6, pp. 1841 --1856, Nov. 1994.

\bibitem{cheng12moment}
L.~Cheng, H.~Ferreira, and I.~Broere, ``Moment balancing templates for
  $(d,k)$-constrained codes and run-length limited sequences,'' \emph{{IEEE}
  Transactions on Information Theory}, vol.~58, no.~4, pp. 2244 --2252, Apr.
  2012.

\bibitem{paluncic12multiple}
F.~Palun\v{c}i\'{c}, K.~Abdel-Ghaffar, H.~Ferreira, and W.~Clarke, ``A multiple
  {Insertion/Deletion} correcting code for run-length limited sequences,''
  \emph{{IEEE} Transactions on Information Theory}, vol.~58, no.~3, pp. 1809
  --1824, Mar. 2012.

\bibitem{levenshtein92perfect}
V.~I. Levenshtein, ``On perfect codes in deletion and insertion metric,''
  \emph{Discrete Mathematics and Applications}, vol.~2, no.~3, pp. 241--258,
  Oct. 1992.

\bibitem{hirschberg00tight}
D.~S. Hirschberg and M.~Regnier, ``Tight bounds on the number of string
  subsequences,'' \emph{Journal of Discrete Algorithms}, vol.~1, no.~1, 2000.

\bibitem{swart03note}
T.~Swart and H.~Ferreira, ``A note on double insertion/deletion correcting
  codes,'' \emph{{IEEE} Transactions on Information Theory}, vol.~49, no.~1,
  pp. 269 -- 273, Jan. 2003.

\bibitem{mercier08number}
H.~Mercier, M.~Khabbazian, and V.~Bhargava, ``On the number of subsequences
  when deleting symbols from a string,'' \emph{{IEEE} Transactions on
  Information Theory}, vol.~54, no.~7, pp. 3279 --3285, Jul. 2008.

\bibitem{liron12characterization}
\BIBentryALTinterwordspacing
Y.~Liron and M.~Langberg, ``A characterization of the number of subsequences
  obtained via the deletion channel,'' \emph{CoRR}, vol. abs/1202.1644, 2012.
  [Online]. Available: \url{http://arxiv.org/abs/1202.1644}
\BIBentrySTDinterwordspacing

\bibitem{levenshtein01efficient}
V.~I. Levenshtein, ``Efficient reconstruction of sequences,'' \emph{{IEEE}
  Transactions on Information Theory}, vol.~47, no.~1, pp. 2 --22, Jan. 2001.

\bibitem{levenshtein01efficientjcomb}
------, ``Efficient reconstruction of sequences from their subsequences or
  supersequences,'' \emph{J. Comb. Theory}, vol.~93, no.~2, pp. 310--332, 2001.

\bibitem{mitzenmacher06polynomial}
M.~Mitzenmacher, ``Polynomial time low-density parity-check codes with rates
  very close to the capacity of the $q$-ary random deletion channel for large
  $q$,'' \emph{{IEEE} Transactions on Information Theory}, vol.~52, no.~12, pp.
  5496 --5501, Dec. 2006.

\bibitem{kanoria09deletion}
\BIBentryALTinterwordspacing
Y.~Kanoria and A.~Montanari. (2009) On the deletion channel with small deletion
  probability. [Online]. Available: \url{http://arxiv.org/abs/0912.5176}
\BIBentrySTDinterwordspacing

\bibitem{diggavi07capacity}
S.~Diggavi, M.~Mitzenmacher, and H.~D. Pfister, ``Capacity upper bounds for the
  deletion channel,'' in \emph{{IEEE} International Symposium on Information
  Theory, 2007. {ISIT} 2007}, Nice, France, Jun. 2007, pp. 1716 --1720.

\bibitem{west00introduction}
D.~B. West, \emph{Introduction to Graph Theory}, 2nd~ed.\hskip 1em plus 0.5em
  minus 0.4em\relax Prentice Hall, Sep. 2000.

\bibitem{sankoff83time}
D.~Sankoff and J.~B. Kruskal, Eds., \emph{\BIBforeignlanguage{en}{Time warps,
  string edits, and macromolecules: the theory and practice of sequence
  comparison}}.\hskip 1em plus 0.5em minus 0.4em\relax Addison-Wesley Pub. Co.,
  Advanced Book Program, 1983.

\bibitem{berge89hypergraphs}
C.~Berge, \emph{Hypergraphs, Volume 45: Combinatorics of Finite Sets},
  1st~ed.\hskip 1em plus 0.5em minus 0.4em\relax North Holland, Aug. 1989.

\bibitem{scheinerman01fractional}
E.~R. Scheinerman and D.~H. Ullman, \emph{Fractional Graph Theory: A Rational
  Approach to the Theory of Graphs}.\hskip 1em plus 0.5em minus 0.4em\relax
  Dover Publications, Dec. 2011.

\bibitem{schrijver98theory}
A.~Schrijver, \emph{Theory of Linear and Integer Programming}.\hskip 1em plus
  0.5em minus 0.4em\relax John Wiley \& Sons, Jun. 1998.

\bibitem{feldman05using}
J.~Feldman, M.~Wainwright, and D.~Karger, ``Using linear programming to decode
  binary linear codes,'' \emph{{IEEE} Transactions on Information Theory},
  vol.~51, no.~3, pp. 954 -- 972, Mar. 2005.

\bibitem{furedi81maximum}
Z.~F\"{u}redi, ``Maximum degree and fractional matchings in uniform
  hypergraphs,'' \emph{Combinatorica}, vol.~1, no.~2, pp. 155--162, 1981.

\bibitem{aharoni96theorem}
R.~Aharoni, R.~Holzman, and M.~Krivelevich, ``On a theorem of {L}ov\'{a}sz on
  covers in $r$-partite hypergraphs,'' \emph{Combinatorica}, vol.~16, no.~2,
  pp. 149--174, Jun. 1996.

\bibitem{marcus98constrained}
B.~Marcus, P.~Siegel, and R.~Roth, ``An introduction to coding for constrained
  systems,'' in \emph{Handbook of Coding Theory}, W.~C. Huffman and V.~Pless,
  Eds.\hskip 1em plus 0.5em minus 0.4em\relax Elsevier, 1998.

\bibitem{marcus98constrainedmono}
\BIBentryALTinterwordspacing
------. (2001) An introduction to coding for constrained systems. [Online].
  Available: \url{http://www.math.ubc.ca/$\sim$marcus/Handbook/index.html}
\BIBentrySTDinterwordspacing

\end{thebibliography}
\end{document}